\documentclass[11pt, a4paper]{article}
\usepackage[textwidth = 7in, textheight=25cm,nohead]{geometry}
\usepackage{ulem}
\usepackage{arydshln}
\usepackage{graphicx, amssymb, amsmath, amsthm, amstext, booktabs, float, multirow, rotating,natbib,bm}
\usepackage{caption}
\usepackage{subcaption}
\usepackage{stackengine}
\usepackage{amssymb}

\usepackage{mathdots}

\let\emptyset\varnothing

\usepackage{tikz}
\usepackage{setspace}
\usepackage{enumerate}
\usepackage{multirow}
\usepackage{graphicx,nicefrac}
\newcommand\bsfrac[2]{%
\scalebox{-1}[1]{\nicefrac{\scalebox{-1}[1]{$#1$}}{\scalebox{-1}[1]{$#2$}}}%
}

\numberwithin{equation}{section} 
\numberwithin{figure}{section} 
\numberwithin{table}{section} 
\newtheorem{definition}{Definition}[section]

\newtheorem{theorem}{Theorem}[section]

\newtheorem{corollary}{Corollary}[section]

\theoremstyle{remark} 
\newtheorem{remark}{Remark}[section] 

\theoremstyle{definition} 
\newtheorem{example}{Example}[section]

\bibpunct{(}{)}{,}{a}{,}{,}

\makeatletter
\newcommand{\doublewidetilde}[1]{{%
\mathpalette\double@widetilde{#1}%
}}
\newcommand{\double@widetilde}[2]{%
\sbox\z@{$\m@th#1\widetilde{#2}$}%
\ht\z@=.9\ht\z@
\widetilde{\box\z@}%
}
\makeatother

\def\ppn{\vskip 6pt \noindent }
\def\R{{\mathbb{R}}}
\def\N{{\mathbb{N}}}
\def\P{{\mathbb{P}}}

\newcommand{{\Xs}}{{\cal X}}
\newcommand{{\Ys}}{{\cal Y}}
\newcommand{{\Ls}}{{\cal L}}
\newcommand{{\Ss}}{{\cal S}}
\newcommand{{\Ms}}{{\cal M}}
\newcommand{{\Gs}}{{\cal G}}
\newcommand{{\Hs}}{{\cal H}}
\newcommand{{\Ns}}{{\cal N}}
\newcommand{{\Is}}{{\cal I}}
\newcommand{{\Vs}}{{\cal V}}
\newcommand{{\Ds}}{{\cal D}}
\newcommand{{\Bs}}{{\cal B}}
\newcommand{{\Cs}}{{\cal C}}
\newcommand{{\Rs}}{{\cal R}}
\newcommand{{\Fs}}{{\cal F}}
\newcommand{{\Us}}{{\cal U}}
\newcommand{{\Ps}}{{\cal P}}
\newcommand{{\ttheta}}{{\bm{\theta}}}
\newcommand{{\Ttheta}}{{\bm{\Theta}}}
\newcommand{{\Oomega}}{{\bm{\Omega}}}
\newcommand{{\Sss}}{{\bm{\Ss}}}
\newcommand{{\pp}}{{\mathbf p}}
\newcommand{{\ww}}{{\mathbf w}}
\newcommand{{\mm}}{{\mathbf m}}
\newcommand{{\uu}}{{\mathbf u}}
\newcommand{{\ppi}}{{\bm{\pi}}}
\newcommand{{\phhi}}{{\bm{\phi}}}
\newcommand{{\pssi}}{{\bm{\psi}}}
\newcommand{{\XX}}{{\mathbf X}}
\newcommand{{\UU}}{{\mathbf U}}
\newcommand{{\BB}}{{\mathbf B}}
\newcommand{{\KK}}{{\mathbf K}}
\newcommand{{\HH}}{{\mathbf H}}
\newcommand{{\II}}{{\mathbf I}}
\newcommand{{\PP}}{{\mathbf P}}
\newcommand{{\yy}}{{\mathbf y}}
\newcommand{{\ee}}{{\mathbf e}}
\newcommand{{\ab}}{{\mathbf a}}

\newcommand{{\dd}}{{\mathbf d}}
\newcommand{{\zero}}{{\mathbf 0}}
\newcommand{{\uno}}{{\mathbf 1}}

\newcommand\indep{\protect\mathpalette{\protect\independenT}{\perp}}
\def\independenT#1#2{\mathrel{\rlap{$#1#2$}\mkern2mu{#1#2}}}

\newcommand{{\toL}}{{\overset{\mathcal{L}}{\longrightarrow}\ }}

\newcommand{{\MC}}{{\,  *_{\text{\scalebox{0.65}{$\Ms$}}}\,  }}

\newcommand{{\dou}}{$\leadsto$\ }

\DeclareMathOperator{\Ran}{Ran}
\DeclareMathOperator{\Supp}{Supp}
\DeclareMathOperator{\Cl}{Cl}

\begin{document}

\setlength{\belowdisplayskip}{5pt} \setlength{\belowdisplayshortskip}{3pt}
\setlength{\abovedisplayskip}{5pt} \setlength{\abovedisplayshortskip}{0pt}

\title{An essay on copula modelling for discrete random vectors;\\  or how to pour new wine into old bottles}
\author{\sc{Gery Geenens}\thanks{Corresponding author: {\tt ggeenens@unsw.edu.au}, School of Mathematics and Statistics, UNSW Sydney, Australia, tel +61 2 938 57032, fax +61 2 9385 7123 }\\School of Mathematics and Statistics,\\ UNSW Sydney, Australia 
}
\date{\today}
\maketitle
\thispagestyle{empty} 


\begin{abstract} Copulas have now become ubiquitous statistical tools for describing, analysing and modelling dependence between random variables. Sklar's theorem, ``the fundamental theorem of copulas'', makes a clear distinction between the continuous case and the discrete case, though. In particular, the copula of a discrete random vector is not identifiable, which causes serious inconsistencies. In spite of this, downplaying statements are widespread in the related literature, and copula methods are used for modelling dependence between discrete variables. This paper calls to reconsidering the soundness of copula modelling for discrete data. It suggests a more fundamental construction which allows copula ideas to smoothly carry over to the discrete case. Actually it is an attempt at rejuvenating some century-old ideas of Udny Yule, who mentioned a similar construction a long time before copulas got in fashion.
\end{abstract}

\section{Introduction}\label{sec:intro}

In \cite{Yule12}, one can read: ``{\it Two association tables that are not directly comparable owing to the different proportions of A's and B's in the data from which the tables were compiled may be rendered directly comparable by multiplying the frequencies in rows and columns by appropriate factors,} [...]  {\it reducing the original tables to some arbitrarily selected standard form}" (p.\,588). The standard form that he recommends is the table whose margins have been made uniform. Likewise, in their extensive study of  association coefficients in $(2 \times 2)$-contingency tables, \citet[p.\,747]{Goodman54} mentioned transforming all marginals to $1/2$ for facilitating interpretation. 
Later, \cite{Mosteller68} developed: ``{\it We might instead think of a contingency table as having a basic} nucleus {\it which describes its association and think of all tables formed by multiplying elements in rows and columns by positive numbers as forming an equivalence class -- a class of tables with the same degree of association}" (p.\,4). And: ``{\it we might especially arrange the table to have uniform margins on each side in the case of a two-way table so as to get a clearer look at the association that is actually occurring}'' (p.\,6).

\ppn If one identifies bivariate discrete distributions with two-way contingency tables, it is clear that what the above historical authors described has much in common with copulas: one tries to capture the dependence structure between the two variables apart from the marginal distributions by making these into uniforms, hence uninformative. The observation is notable, as it has been known at least since \cite{Marshall96} that the notion of copula fits poorly in the discrete framework. Here `copula' refers to the classical definition \citep[Definition 1.3.1]{Durante16} which, in the bivariate case, reads:
\begin{definition}\label{dfn:classcop} A bivariate copula $C$ is a function from $\Is \doteq [0,1]^2$ to $[0,1]$ defined as
\begin{equation*}
C(u,v) = \P(U \leq u, V \leq v), 
\end{equation*}
where $U,V \sim \Us_{[0,1]}$, the continuous uniform distribution on the unit interval.
\end{definition}
Such copulas naturally arise in statistical modelling through the celebrated Sklar's theorem \citep{Sklar59}:
\begin{theorem}[{\bf Sklar}] \label{thm:Sklar} Let $F_{XY}$ be the distribution function of a bivariate random vector $(X,Y)$, with marginal distribution functions $F_X$ and $F_Y$. Then there exists a copula $C$ such that, for all $(x,y) \in \R^2$,
\begin{equation} F_{XY}(x,y) = C(F_X(x),F_Y(y)). \label{eqn:Sklar} \end{equation}
If $F_X$ and $F_Y$ are continuous, then $C$ is unique; otherwise $C$ is uniquely determined on $\Ran F_X \times \Ran F_Y$ only. Conversely, for any univariate distribution functions $F_X$ and $F_Y$ and any copula $C$, the function $F_{XY}$ defined by (\ref{eqn:Sklar}) is a valid bivariate distribution function with
marginals $F_X$ and $F_Y$.
\end{theorem}
The popularity of copulas for dependence modelling largely follows from quotes like `{\it Copulas allow us to separate the effect of dependence from effects of the marginal distributions}'. Clearly, if $C$ is unique, then it unequivocally characterise how the two marginals $F_X$ and $F_Y$ interlock for producing the joint behaviour of $(X,Y)$, while being ignorant of what those marginals are. For instance, if $X$ and $Y$ are independent ($X \indep Y$) {\it and} if $C$ is unique, then from (\ref{eqn:Sklar}) $C$ must be the `product copula'
\begin{equation} \Pi(u,v) = uv \qquad \forall (u,v) \in \Is, \label{eqn:indepcop} \end{equation}
and this regardless of $F_X$ and $F_Y$. It is often overlooked that the situation is this appealing only in the case of continuous margins, when there is one-to-one correspondence between the joint distribution $F_{XY}$ and its copula $C$. Without that bijectivity, i.e., for $X$ and/or $Y$ discrete, the above argument falls apart.  

\ppn Instrumental to copula ideas is the distribution of the vector $(F_X(X),F_Y(Y))$. If $X$ and $Y$ are both continuous, then, through `{\it Probability Integral Transform}' (PIT), $F_X(X)$ and $F_Y(Y)$ have uniform distributions $\Us_{[0,1]}$, and the copula $C$ is their joint distribution. Clearly one can plug any increasing transformations of $X$ and/or $Y$ into PIT with the same output. Hence copulas are invariant under increasing transformations of the margins \citep[Theorem 2.4.3]{Nelsen06}, that is, `margin-free'. Any copula-based dependence measure, such as Kendall's or Spearman's correlations \citep[Chapter 5]{Nelsen06}, is then `margin-free' as well.

\ppn Now, in the case $X$ and/or $Y$ discrete, $\Ran F_X$ and/or $\Ran F_Y$ are just countable subsets of $[0,1]$. The distributions of $F_X(X)$ and/or $F_Y(Y)$ are thus not $\Us_{[0,1]}$, and their joint distribution cannot be a copula as described by Definition \ref{dfn:classcop}. It is actually a {\it subcopula}, i.e., a function satisfying the main structural properties of copulas but whose support is only a strict subset of $\Is$ containing 0 and 1 \citep[Definition 2.2.1]{Nelsen06}. Any such subcopula can be extended into a copula \citep[Lemma 2.3.5]{Nelsen06}: the gaps in $\Is \ \backslash (\Ran F_X \times \Ran F_Y)$ can be filled in a way preserving the properties of copulas;  however there are uncountably many ways of doing so and $C$ in (\ref{eqn:Sklar}) is not identifiable. 

\ppn Unidentifiability of $C$ {\it does} cause serious inconsistencies. \cite{Marshall96} was the first to list some, while later \cite{Genest07b} systematically investigated them and painted a rather depressing picture of the situation. Though, they concluded on a note of hope: ``{\it copula-based models are likely to become as attractive for discrete variables as they have grown to be for continuous data}". Here, we must share with \cite{Faugeras17} a much less positive view about the soundness of copula modelling for discrete data -- see Section \ref{sec:copfordiscr}. However, we might eventually agree with \cite{Genest07b}'s final word if the concept of `copula' was given a more fundamental meaning, not limited to Definition \ref{dfn:classcop} but agreeing with it in the continuous case. This paper precisely investigates such a construction. Actually it is an attempt at rejuvenating Yule's, Goodman and Kruskal's and Mosteller's conceptions, to make them fit into some modern `extended copula modelling methodology'. 

\section{Copulas on discrete distributions} \label{sec:copfordiscr}

It is fair to say that all the reasons which make copulas attractive and effective for modelling dependence in the continuous case, break up in the discrete case: ``{\it everything that can go wrong, will go wrong}" \citep[p.\,641]{Embrechts09}. The concluding positive feeling of \cite{Genest07b} probably follows mostly from their Example 13 of a bivariate Bernoulli distribution $F_{XY}$. They showed that consistent estimation of the parameter of a postulated Farlie-Gumbel-Morgenstern (FGM) copula \citep[Example 3.12]{Nelsen06} on $F_{XY}$ was possible and provided a reasonable description of the dependence structure of the underlying discrete random vector. 

\ppn Recently, though, that example was picked apart in \cite{Faugeras17}, who described how the FGM copula is compatible with a bivariate Bernoulli distribution only for some values of the parameters of the univariate Bernoulli marginals, but not for others. Here appears clearly that, in the discrete case, one can never detach the copula from the marginals. 
The fact that the copula-based measures of dependence (e.g.\ Kendall's or Spearman's) are margin-dependent was already observed in \citet[Proposition 2.3]{Marshall96} and \citet[Section 4.2]{Genest07b}, but what \cite{Faugeras17} describes goes well beyond that: the copula model {\it per se} may or may not be intrinsically meaningful depending on the marginals. A similar observation was made earlier in \citet[Section 1.1]{Zilko16}, although this was not seen as a problem there.

\ppn The following extension of Example 5 in \cite{Genest07b} is another compelling example of the inadequacy of copulas for modelling dependence between discrete variables. Suppose that $X \sim \text{Bern}(\pi_X),Y \sim \text{Bern}(\pi_Y)$ for two probabilities $\pi_X, \pi_Y \in (0,1)$, and $ X \indep Y$. Then, for reconstructing the corresponding bivariate Bernoulli $F_{XY}$ it is enough to plug in (\ref{eqn:Sklar}) any copula $C$ such that 
\begin{equation} C(1-\pi_X,1-\pi_Y)= (1-\pi_X)(1-\pi_Y). \label{eqn:constrindep} \end{equation}
This is easily seen by inspection, but this is confirmed directly by Sklar's theorem: $C$ is only identifiable on $\Ran F_X \times \Ran F_Y = \{0,1-\pi_X,1\} \times \{0,1-\pi_Y,1\}$, but given that the behaviour of $C$ along the sides of $\Is$ is fixed by trivial constraints (uniform margins), only what happens at $(1-\pi_X,1-\pi_Y)$ brings valuable information. The product copula (\ref{eqn:indepcop}) naturally fulfils (\ref{eqn:constrindep}), but so does a wide spectrum of other copulas of miscellaneous shapes whose only common trait is to go through $\left(1-\pi_X,1-\pi_Y,(1-\pi_X)(1-\pi_Y)\right) \in (0,1)^3$. One can legitimately question any conclusion drawn from such a model: the element supposed to describe the dependence structure, i.e.\ $C$, may interchangeably characterise independence or dependence of various strength and nature. In particular, any dependence measure based on the fitted copula is uninterpretable, given that the fitted copula could characterise drastically different dependence structures.


\ppn In consequence, it seems difficult not to controvert downplaying statements commonly found in the related literature, alleging that unidentifiability does not preclude the use of parametric copulas for modelling discrete data. 
Admittedly, one can always take two univariate discrete distributions and bind them together through a copula $C$ that we have picked; the `{\it Conversely}'-part of Sklar's theorem guarantees that this produces a valid bivariate discrete distribution with the right marginals. But there is actually no special link between the `input' copula $C$ and the `output' distribution $F_{XY}$. For instance, \cite{Faugeras17} explains how the bivariate Bernoulli distribution on which \citet[Example 13]{Genest07b} fitted a FGM copula, say $C_\text{FGM}$, could have been obtained all the same from a Plackett or an Ali-Mikhail-Haq copula, or from the reader's `{\it peculiar favourite copula family}' \citep[p.\,128]{Faugeras17}. It is enough to fix the parameter(s) of the copula so as to make it go through the `magical point' $\left(1-\pi_X,1-\pi_Y,C_{\text{FGM}}(1-\pi_X,1-\pi_Y)\right)$, making it futile to mention any parametric copula model at all in this case. 

\section{Transformations of the margins to uniforms} \label{sec:transunif}

The root of all trouble is that the usual PIT result, $F_X(X) \sim \Us_{[0,1]}$, does not hold for $X$ discrete. Though, the $\Us_{[0,1]}$-distribution of $F_X(X)$ and $F_Y(Y)$ in the continuous case is clearly what prompted Definition \ref{dfn:classcop}, and the induced widespread belief that copula methods are based on transformations of the margins into uniforms. Thus the main idea behind copulas, and even the very definition of a copula, are unfit for the discrete framework, reinforcing the feeling that any attempt at modelling dependence between non-continuous variables based on such classical copulas is doomed to failure from the outset. 

\ppn Clearly, the only way one can transform a discrete random variable into a continuous uniform is to resort to some sort of randomisation. Hence, to make the discrete case forcibly fit into the continuous copula framework, a common practice has been to appropriately `jitter' the original discrete variables with some uniform random noise. The so-created artificial continuous random vector has a unique copula, known as the {\it checkerboard} copula $C^\maltese$. Arguably, $C^\maltese$ retains some of the dependence structure of the original discrete vector \citep{Schweizer74,Denuit05,Genest07b,Neslehova07
}, and is a valid copula extension of the underlying subcopula \citep{Faugeras15}. However, $C^\maltese$ is just a particular choice -- and not always the most natural one -- among all the copulas satisfying (\ref{eqn:Sklar}), and by itself does not solve any of the problems exposed above. 


\ppn Now, in his stance against copulas, \citet{Mikosch06} explicitly asked (his Section 4) `{\it Why does one transform the marginals to a uniform distribution?}', and failing to come up with any compelling mathematical answer (among other things) lead him to reject the idea of copulas altogether. Yet, it has been widely acknowledged since then \citep{Embrechts09}, but even long before \citep[p.\,69]{Hoeffding40}, that the choice of transforming the margins to uniforms is driven by convenience only
. Given that transforming to uniform is precisely the stumbling block of copula methods for discrete variables, one may sensibly ask: why stick to an inessential choice initially made for convenience only, if it is no more convenient at all in the situation of interest?

\ppn Indeed forcing uniform marginals necessarily requires `guessing' what the suitable copula $C$ might be beyond $\Ran F_X \times \Ran F_Y$, and it is not clear what is the value of such guesswork. Sklar's theorem establishes that the `interesting values' for comprehending the joint behaviour of $(X,Y)$ coincide with some incidental copula $C$ evaluated on $\Ran F_X \times \Ran F_Y$. A naive interpretation of this puts the element $C$ in the foreground whereas it is actually of little importance. \cite{Vapnik98} famously wrote: `{\it one should avoid solving more difficult intermediate problems when solving a target problem}'. Here we should directly focus on those `interesting values' instead of playing a guessing game that has no definite answer anyway, as $C$ is not identifiable. In other words, there is no reason to extend the unique subcopula of a discrete bivariate vector to a copula, and any justifiable analysis of the underlying dependence structure should be undertaken at the subcopula level, or equivalent. 

\ppn What this means concretely is clear when looking again at the bivariate Bernoulli example. In this case, only the value of $C$ at $(1-\pi_X,1-\pi_Y)$ contains valuable information (see the lines following (\ref{eqn:constrindep})). So, what Sklar's theorem {\it fundamentally} says is that the whole dependence can be captured by one single number. Of course this is directly backed up by any basic analysis of the bivariate Bernoulli distribution as a $(2 \times 2)$-contingency table. Form the probability mass function (pmf) $\P(X=x,Y=y) \doteq p_{xy}$, $(x,y) \in \{0,1\} \times \{0,1\}$, into a table with 2 rows and 2 columns, such as (\ref{eqn:bivbernpmf}) below. Given that $\sum_{x,y} p_{xy} = 1$, the number of degrees of freedom of such a table is 3, one of these being used when fixing the first margin $\pi_X = p_{10}+p_{11}$, another one when fixing the other $\pi_Y = p_{01}+p_{11}$. So only one degree of freedom stays for describing what remains once the marginals are known, that is, the level of association in the table -- cf.\ the $\chi^2$-test of independence. It is not clear what would be the benefit of playing on a whole bivariate function $C$ when only one single number contains all the required information.

\ppn \citet{Edwards63} argued that this number should be the odds-ratio 
\begin{equation} \omega = \frac{p_{00}p_{11}}{p_{10}p_{01}} \label{eqn:OR} \end{equation} 
(or any monotonic function thereof) because it is `margin-free' (he did not use that term, though, but see his Corollary 2). It will be shown in Section \ref{sec:Berncop} that there is indeed a one-to-one correspondence between $\omega$ and the value $C(1-\pi_X,1-\pi_Y) = \P(X\leq 0,Y\leq 0)=p_{00}$ singled out by Sklar's theorem in this situation. This means that using $\omega$ as single dependence parameter is in total agreement with Sklar's theorem: we might look at $\omega$ on another scale to make it match $C(1-\pi_X,1-\pi_Y)$, hence to agree with the subcopula. 
It is also a simple algebraic exercise (Section \ref{sec:arbbern}) to show that, given the margins, the full table (i.e., the bivariate pmf) can be reconstructed from the value of $\omega$ only. Hence the marginal distributions {\it coupled} with the margin-free dependence parameter $\omega$ unequivocally defines the bivariate distribution of interest. Clearly, the single number $\omega$ entirely fulfils what we would like the role of a copula to be, while by no means being related to Definition \ref{dfn:classcop}.

\ppn Transformation to uniform marginals is thus clearly not a necessary step for making sense of the main ideas behind copula modelling. Indeed, in Section \ref{sec:copclass}, an alternative perspective on copulas is given, not relying explicitly on PIT. Avoiding PIT allows the concept to be readily adapted to the discrete case as well, while keeping all the pleasant properties of usual copula modelling, in particular `margin-freeness' of any copula-based quantities.

\section{Copulas as equivalence classes of dependence}  \label{sec:copclass}

Let $(X,Y)$ be a continuous vector with distribution $F_{XY}$. For simplicity, assume\footnote{This is not restrictive, one can imagine that we observe $X \in \R$ and $Y \in \R$ on the inverse logit scale, for instance, and copulas are invariant to monotonic transformations of the margins in any case.}  that $X$ and $Y$ are both supported on $[0,1]$ and that $F_{XY}$ admits a density $f_{XY}$ with marginal densities $f_X$ and $f_Y$ on the unit square $\Is$. Let $\Fs = \{f:\Is \to \R, \text{ s.t. } f \geq 0, \iint_\Is f = 1 \}$, the set of all bivariate probability densities on $\Is$, and $\Ss$ the set of all differentiable strictly increasing functions from $[0,1]$ to $[0,1]$. See that $(\Ss,\circ)$, where $\circ$ denotes function composition, is a group, and so is $(\Ss \times \Ss, \circ .)$, where $\circ .$ denotes componentwise composition: for $(\Phi_1, \Psi_1),(\Phi_2, \Psi_2) \in \Ss \times \Ss$, $(\Phi_1, \Psi_1) \circ. (\Phi_2, \Psi_2) = (\Phi_1 \circ \Phi_2, \Psi_1 \circ \Psi_2)$.

\ppn For any $(\Phi, \Psi) \in \Ss \times \Ss$, define $g_{\Phi,\Psi}: \Fs \to \Fs$ as
\begin{equation} g_{\Phi,\Psi}(f)(u,v) =   \frac{f(\Phi^{-1}(u),\Psi^{-1}(v))}{|\Phi'(\Phi^{-1}(u))|\,|\Psi'(\Psi^{-1}(v))|} . \label{eqn:margdist} \end{equation}
Now, for $(\Phi_1, \Psi_1),(\Phi_2, \Psi_2) \in \Ss \times \Ss$, it can be seen that 
\[\left(g_{\Phi_2, \Psi_2} \circ g_{\Phi_1, \Psi_1}\right)(f) = g_{\Phi_2 \circ \Phi_1,\Psi_2 \circ \Psi_1}(f) \]
(compatibility), while if $\Phi(x) = x$ and $\Psi(y) = y$, then $g_{\Phi,\Psi}(f) = f$ (identity). This makes $g_{\Phi,\Psi}$ a group action \citep[Section 10.1]{Eie10} of $(\Ss \times \Ss, \circ .)$ on $\Fs$, which defines {\it orbits}: for any $f \in \Fs$, let $[f] = \{f^* \in \Fs: \exists (\Phi,\Psi) \in \Ss \times \Ss \text{ s.t. } f^* = g_{\Phi,\Psi}(f)\}$. Such orbits induce an equivalence relation, viz.\ $f_1 \sim f_2$ whenever $[f_1] = [f_2]$. The quotient space $\overline{\Fs} = \Fs /(\Ss \times \Ss, \circ .)$ is the set of all equivalence classes. 

\ppn From standard arguments on transformation of random variables, $g_{\Phi,\Psi}(f_{XY})$ is the joint density of $(\Phi(X),$ $\Psi(Y))$, so essentially the version of $f_{XY}$ whose marginal distributions have been individually distorted by $\Phi$ and $\Psi$. The class $[f_{XY}]$ contains all those `marginally distorted' densities which share the same core as $f_{XY}$. Free from any sense of margins, that core can only be what remains between the margins, that is, the `glue' between the margins inside $f_{XY}$. According to \cite{Tankov15}, this is the exact definition of `dependence': `{\it the information on the law of a random vector which remains to be determined once the marginal laws of its components have been specified.}' Each equivalence class in $\overline{\Fs}$ is thus representative of a certain {\it dependence structure}. 

\ppn Arguably, the elements of $\overline{\Fs}$ are really the objects which deserve the name `copula', as they genuinely {\it are} the links (`{\it copulae}' in Latin) which cement marginals inside bivariate densities. However, to avoid any confusion with the classical Definition \ref{dfn:classcop}, we will call the element $[f] \in \overline{\Fs}$ the `nucleus' of $f$ to align with \cite{Mosteller68}'s description (Section \ref{sec:intro}). 
A nucleus $[f]$ is, in some sense, akin to a bivariate density which has been entirely stripped from its marginals. Of course, with no marginals, $[f]$ in itself is not a density.



\ppn Precisely, the abstract concept of a bivariate density with no marginals is difficult to visualise. Hence, for describing the inner dependence structure of the vector $(X,Y)$, one may want to exhibit a simple re-embodiment of $[f_{XY}]$ into a proper density by gluing back on it some default marginals. The simplest choice for those default marginals seems to be uniform densities. By PIT, this, of course, is the element $g_{\Phi,\Psi}(f_{XY}) \in [f_{XY}]$ which corresponds to $(\Phi,\Psi) = (F_X,F_Y)$. That particular representative is thus
\[\bar{f}_{XY}(u,v) = \frac{f_{XY}(F_X^{-1}(u),F_Y^{-1}(v))}{f_X(F_X^{-1}(u))f_Y(F_Y^{-1}(v))} \in [f_{XY}], \]
in which we recognise the density $c$ of the copula $C$ of $F_{XY}$ described by Theorem \ref{thm:Sklar}.

\ppn As stressed in Section \ref{sec:transunif}, the choice of uniform margins for re-embodying $[f_{XY}]$ into a proper density is totally arbitrary. It seems just sensible, for interpretation and visualisation purpose, to keep things as uncomplicated as possible, and the uniform distribution is arguably the simplest choice. 
That said, uniforms and/or PIT do not play any role when defining the concept of nucleus, which is really what copulas are all about. The construction of such nuclei can thus be adapted {\it mutatis mutandis} to discrete distributions. The process is detailed for the case of a bivariate Bernoulli distribution in the next section, and generalised to other bivariate discrete distributions after that.

\section{The Bernoulli copula} \label{sec:Berncop}

\subsection{The bivariate Bernoulli distribution}

Consider again the case of two Bernoulli random variables $X \sim \text{Bern}(\pi_X)$ and $Y \sim \text{Bern}(\pi_Y)$ sharing (potentially) some dependence. The corresponding bivariate Bernoulli distribution, say $\pp$, is typically  presented under the form of a $(2 \times 2)$-table, such as
\begin{equation} \begin{array}{c l ||c c c c | c}
& \bsfrac{Y}{X} & & 0 & 1 & \\
\hline\hline
& 0 & & p_{00} & p_{01}& & p_{0\bullet} \\
& 1 & & p_{10}  & p_{11}& & p_{1\bullet}  \\
\hline 
& & & p_{\bullet 0} & p_{\bullet 1}& & 1 \\
\end{array}, \label{eqn:bivbernpmf} \end{equation}
where for $x,y \in \{0,1\}$, $p_{xy} = \P(X=x,Y=y)$, $p_{\bullet y} = p_{0y} + p_{1y}$ and $p_{x \bullet} = p_{x 0} + p_{x 1}$. Of course, $\pi_X=p_{1\bullet}$ and $\pi_Y = p_{\bullet 1}$. Assume $0 < \pi_X <1$ and $0<\pi_Y<1$ (non-degenerate table). Define $\Ps_{2 \times 2}$ the set of all such bivariate Bernoulli probability mass functions, where each $\pp \in \Ps_{2\times 2 }$ is identified to the matrix
\begin{equation} \pp=\begin{pmatrix} p_{00} & p_{01} \\ p_{10} & p_{11} \end{pmatrix}. \label{eqn:matBivBern} \end{equation}


\ppn Now, as $\sum_{x,y} p_{xy} = 1$, one can actually identify $\Ps_{2 \times 2}$ to the 3-dimensional simplex, here a regular tetrahedron whose vertices are the degenerate distributions 
\[\dd_1 = \begin{pmatrix} 1 & 0 \\ 0 & 0 \end{pmatrix},\  \dd_2= \begin{pmatrix} 0 & 1 \\ 0 & 0 \end{pmatrix},\   \dd_3 = \begin{pmatrix} 0 & 0 \\ 1 & 0 \end{pmatrix}, \text{ and } \dd_4 = \begin{pmatrix} 0 & 0 \\ 0 & 1 \end{pmatrix};  \]
see Figure \ref{fig:tetra}. We will call this tetrahedron the {\it Bernoulli tetrahedron}. Note that, as we assume $0 < \pi_X,\pi_Y <1$, the 4 vertices and the edges $\dd_1 \dd_2$, $\dd_1 \dd_3$, $\dd_2 \dd_4$ and $\dd_3 \dd_4$ are not  admissible elements of $\Ps_{2\times 2}$. 

\begin{figure}[h]
\centering
\includegraphics[width=0.5\textwidth,trim=0 2cm 0 0.8cm]{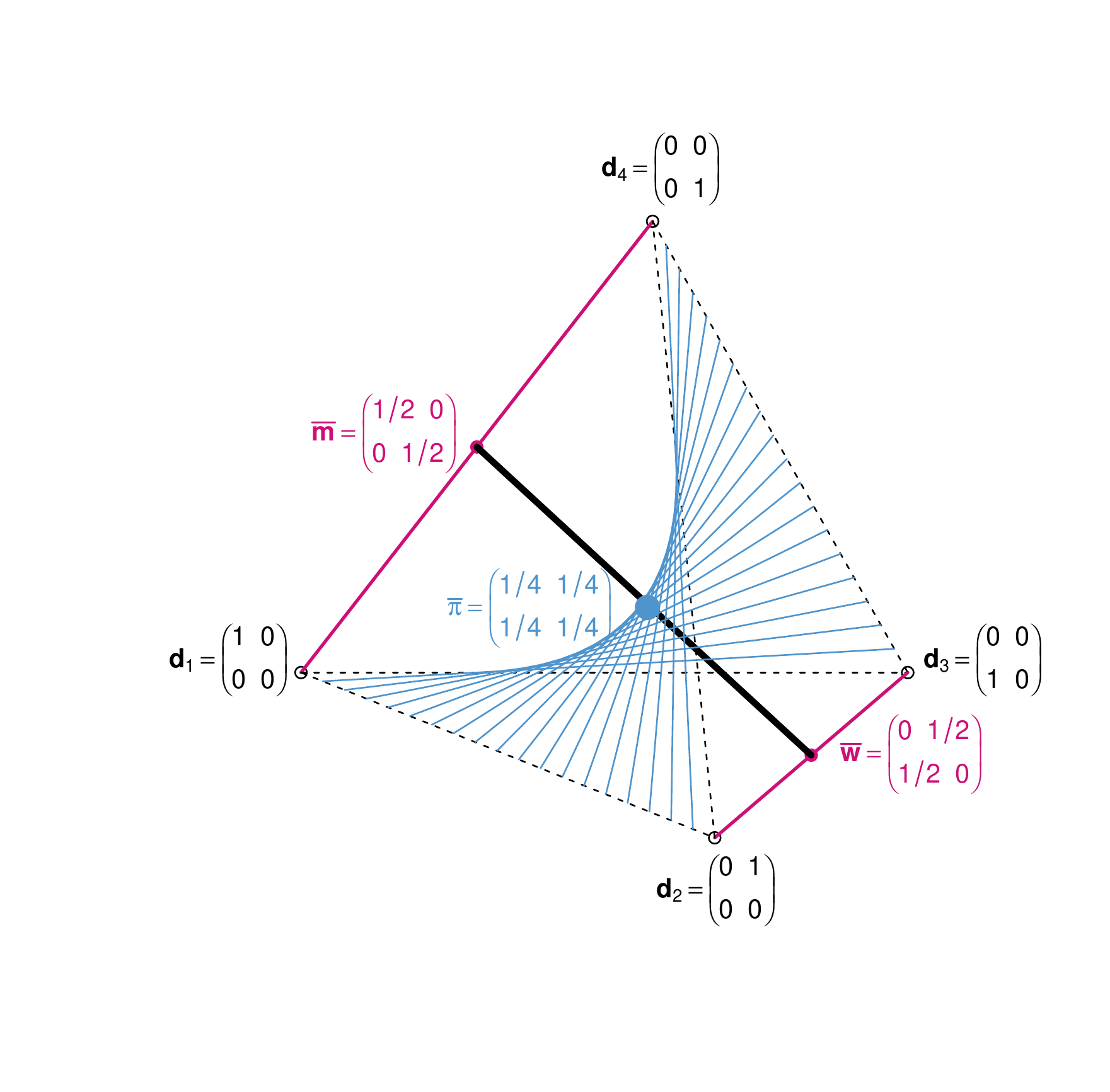}
\caption{The Bernoulli tetrahedron. The blue surface is the Bernoulli nucleus of independence. The thick line $\overline{\ww}\, \overline{\mm}$ is the locus of all the Bernoulli copula pmf's. The centre of the tetrahedron is the Bernoulli independence copula $\overline{\ppi}$. The magenta edges represent `absolute association' (positive, edge $\dd_1 \dd_4$; negative, edge $\dd_2 \dd_3$). Their mid-points $\overline{\mm}$ and $\overline{\ww}$ are the upper and lower Fr\'echet bounds for Bernoulli copulas.}
\label{fig:tetra}
\end{figure}

\subsection{Marginal transformations} \label{subsec:margtrans}

Following Section \ref{sec:copclass}, one would like to extract the core of $\pp$, i.e., what remains invariant to `monotonic distortion' of the margins. In (\ref{eqn:margdist}), by such distortion it is meant the vector $(\Phi(X),\Psi(Y))$ where $\Phi$ and $\Psi$ are increasing functions. If the continuous variable $X$ has density $f_X(u)$ at $u \in [0,1]$, then $\Phi(X)$ has density  
\begin{equation} f_{\Phi(X)}(u) = \frac{f_X(\Phi^{-1}(u))}{\Phi'(\Phi^{-1}(u))},  \label{eqn:distdens} \end{equation}
which can take any arbitrary shape depending on $\Phi$ (and similar for $\Psi(Y)$). By contrast, the `transformation trick' does not work for discrete random variables. In particular, for $X \sim \text{Bern}(\pi_X)$, $\Phi(X)$ remains a two-point distribution with the exact same ratio $(1-\pi_X,\pi_X)$ (only the `labels' change). However, one can see (\ref{eqn:distdens}) from a more basic perspective, considering $\Phi$ as just a mechanism which re-assigns the initial probability mass differently. Under the effect of $\Phi$, the value $u \in [0,1]$, initially assigned the probability $f_X(u)\,du$, would now get a probability $f_X^*(u)\,du$ where $f_X^*(u)$ is given by (\ref{eqn:distdens}). Note that the factor $1/\Phi'(\Phi^{-1}(u))$ is just a normalisation guaranteeing that $\int f_X^*(u)\,du = 1$. Now this more fundamental interpretation of (\ref{eqn:distdens}) carries over to the Bernoulli framework.

\ppn Indeed, for some $\phi > 0$, define a distorted distribution for $X$ as Bern$(\pi^*_X)$, where $\pi^*_X = \frac{\phi\pi_X}{1-\pi_X + \phi \pi_X}$. Clearly, for $\phi > 1$, $\pi^*_X > \pi_X$: some of the probability initially assigned to $X=0$ has been transferred to the next value, $X=1$; and reversely for $\phi<1$. Like above, the factor $1/(1-\pi_X + \phi \pi_X)$ in $\pi_X^*$  is just a normalisation, guaranteeing $\pi^*_X \in [0,1]$ for all $\phi>0$. The margin $Y$ can be similarly distorted. When both marginal distributions are simultaneously perturbed in that way, the initial joint probability distribution is re-assigned through table (\ref{eqn:bivbernpmf}) by a similar process of transferring probability weight between adjacent cells. The organisation of the cells, in particular their order along each margin, is not altered: the marginal distortions are monotonic in that sense. In effect, the resulting distorted table is obtained by multiplying the rows and columns of (\ref{eqn:matBivBern}) by positive values (and renormalise). This totally concords with what \cite{Yule12} and \cite{Mosteller68} urged; see Section 1.

\ppn Specifically, define $\Ds^{(1)}_{2 \times 2}$ the set of all diagonal matrices whose entry $(1,1)$ is equal to 1, and for any $\phi,\psi>0$, set
\[\Phi =  \begin{pmatrix} 1 & 0 \\ 0 & \phi \end{pmatrix} \in \Ds^{(1)}_{2 \times 2} \quad \text{ and } \quad \Psi =  \begin{pmatrix} 1 & 0 \\ 0 & \psi \end{pmatrix} \in \Ds^{(1)}_{2 \times 2}. \]
Of course $(\Ds^{(1)}_{2 \times 2}, \cdot)$, where $\cdot$ is matrix multiplication, is a group (with $I$ its identity), and so is $(\Ds^{(1)}_{2 \times 2} \times \Ds^{(1)}_{2 \times 2},\cdot .)$, where $\cdot .$ is componentwise matrix multiplication: for $(\Phi_1, \Psi_1),(\Phi_2, \Psi_2) \in (\Ds^{(1)}_{2 \times 2} \times \Ds^{(1)}_{2 \times 2})$, $(\Phi_1, \Psi_1) \cdot. (\Phi_2, \Psi_2) = (\Phi_1 \cdot \Phi_2, \Psi_1 \cdot \Psi_2)$. Now define $g_{\Phi,\Psi}: \Ps_{2 \times 2} \to \Ps_{2 \times 2}$:
\begin{equation} g_{\Phi,\Psi}(\pp) =  \frac{ \Phi \cdot \pp \cdot \Psi}{\| \Phi \cdot \pp \cdot \Psi\|_1} = \frac{1}{p_{00} + \psi p_{01} + \phi p_{10} + \phi \psi p_{11} } \begin{pmatrix} p_{00} & \psi p_{01} \\ \phi p_{10} & \phi \psi p_{11} \end{pmatrix}. \label{eqn:Berntrans} \end{equation} 
Similarly to Section \ref{sec:copclass}, $g_{\Phi,\Psi}$ is a group action of $(\Ds^{(1)}_{2 \times 2} \times \Ds^{(1)}_{2 \times 2},\cdot .)$ on $\Ps_{2\times 2}$. Any $\pp\in \Ps_{2\times 2}$ induces an orbit $[\pp] = \{\pp^* \in \Ps_{2\times 2} : \exists (\Phi,\Psi) \in \Ds^{(1)}_{2 \times 2} \times \Ds^{(1)}_{2 \times 2} \text{ s.t. } \pp^*=g_{\Phi,\Psi}(\pp) \}$. Those orbits form equivalence classes of bivariate Bernoulli distributions: $\pp_1 \sim \pp_2$ whenever $[\pp_1] = [\pp_2]$. The quotient space $\overline{\Ps}_{2\times 2} = \Ps_{2\times 2} / (\Ds^{(1)}_{2 \times 2} \times \Ds^{(1)}_{2 \times 2},\cdot .)$ is the set of all those equivalence classes. 

\begin{remark} \label{rmk:cda} Identifying $\Ps_{2\times 2}$ to the tetrahedron allows a parallel with compositional data analysis, i.e., data living on the simplex. Indeed the transformation (\ref{eqn:Berntrans}) can be written
\[\begin{pmatrix} \frac{1}{2(1+\phi)} & \frac{1}{2(1+\phi)} \\ \frac{\phi}{2(1+\phi)} & \frac{\phi}{2(1+\phi)} \end{pmatrix} \oplus \pp \oplus \begin{pmatrix} \frac{1}{2(1+\psi)} & \frac{\psi}{2(1+\psi)} \\ \frac{1}{2(1+\psi)} & \frac{\psi}{2(1+\psi)} \end{pmatrix} \doteq \phhi \oplus \pp \oplus \pssi
, \]
where $\oplus$ is the `perturbation' operator \citep{Aitchison01}, arguably the most natural operation on the simplex. See that $\phhi \in \Ps_{2 \times 2}$ with $X \sim \text{Bern}(\phi/(1+\phi))$ and $Y\sim \text{Bern}(1/2)$, while $\pssi \in \Ps_{2 \times 2}$ with $X\sim \text{Bern}(1/2)$ and $Y\sim \text{Bern}(\psi/(1+\psi))$. See also that in both distributions $\phhi$ and $\pssi$, $X$ and $Y$ are independent. Clearly, $\phhi$ aims at distorting solely the margin $X$, $\pssi$ aims at distorting solely the margin $Y$, but none is allowed to bring extra dependence into $\pp$. This is formalised in the next section. \qed
\end{remark}

\subsection{Bernoulli nucleus and Bernoulli copula probability mass function} \label{subsec:Berncoppmf}

\ppn Any $[\pp] \in \overline{\Ps}_{2\times 2}$ can thus be interpreted as the class of all bivariate Bernoulli distributions (\ref{eqn:matBivBern}) which share the same `core' structure once we strip them from their marginal distributions. 
This suggests that the equivalence classes may again be classes of equivalent dependence, which is directly confirmed by noting that the odds-ratio (\ref{eqn:OR}) is class-invariant. Specifically, for a distribution $\pp \in \Ps_{2 \times 2}$, define  
\[\omega: \Ps_{2 \times 2} \to \R^+: \omega(\pp) = \frac{p_{00} p_{11}}{p_{10}p_{01}} . \]
It is obvious that, for any $(\Phi,\Psi) \in \Ds^{(1)}_{2 \times 2} \times \Ds^{(1)}_{2 \times 2}$,
\begin{equation} \omega(g_{\Phi,\Psi}(\pp)) = \frac{p_{00} \phi \psi p_{11}}{\phi p_{10} \psi p_{01}} = \frac{p_{00} p_{11}}{p_{10}p_{01}} = \omega(\pp). \label{eqn:comOR} \end{equation}
As in the previous Section, we call $[\pp]$ the nucleus of $\pp$ (although it would probably deserve the name of `copula' as well), as it contains nothing else but the information of how the Bernoulli marginals are glued together inside $\pp$. Hence the quotient space $\overline{\Ps}_{2\times 2}$ forms the family of {\it Bernoulli nuclei}.

\ppn \citet[Section 3]{Fienberg70b} showed that, in the Bernoulli tetrahedron, the sets of distributions $\pp \in \Ps_{2\times 2}$ sharing common odds ratios $\omega \in (0,\infty)$ are doubly-ruled surfaces corresponding to sections of hyperboloids of one sheet. E.g., Figure \ref{fig:tetra} shows the surface corresponding to $\omega =1$, that is, all bivariate Bernoulli distributions for which $X \indep Y$. Running $\omega$ over $(0,\infty)$ produces similar non-intersecting surfaces, which are the Bernoulli nuclei $[\pp] \in \overline{\Ps}_{2\times 2}$. For the limiting cases $\omega=0$ and $\omega = \infty$, see Section \ref{sec:structzeros}.

\ppn The fact that each Bernoulli nucleus $[\pp]$ is described by its odds-ratio $\omega(\pp)$ makes it easy to get a sense of the dependence involved. However, it may still be insightful to define a representative of $[\pp]$, that is, a particular `simple' bivariate Bernoulli distribution with odds-ratio $\omega(\pp)$. Again, a natural choice seems to be the element of $[\pp]$ with uniform margins, as \cite{Goodman54} suggested (Section \ref{sec:intro}). 

\ppn Simple algebra reveals that, for $\pp \in \Ps_{2\times 2}$ such that $\omega(\pp) =  \omega \geq 0$, there is a unique element in $[\pp]$ with uniform margins, which is
\begin{equation} \overline{\pp} = \begin{pmatrix} \frac{\sqrt{\omega}}{2(1+\sqrt{\omega})} & \frac{1}{2(1+\sqrt{\omega})} \\ \frac{1}{2(1+\sqrt{\omega})} & \frac{\sqrt{\omega}}{2(1+\sqrt{\omega})} \end{pmatrix}. \label{eqn:Berncoppmf} \end{equation}
Naturally, here, the margins are {\it discrete} uniforms. Essentially `margin-free', one can think of $\overline{\pp}$ as a distribution on any appropriate $(2 \times 2)$-points domain. Although not essential, one can agree that $\overline{\pp}$ is a distribution on $\{\frac{1}{3},\frac{2}{3}\} \times \{\frac{1}{3},\frac{2}{3}\}$, so as to stay mostly aligned with the usual `uniform on $[0,1]$' copula specification. This particular choice will have interesting implications in Section \ref{sec:copinfsup}. The representative (\ref{eqn:Berncoppmf}) is akin to the copula density in the continuous case. Hence we call $\overline{\pp}$ the {\it Bernoulli copula probability mass function} (copula pmf). Note that the values in (\ref{eqn:Berncoppmf}) were mentioned in \citet[11.2-14]{Bishop75}, while a similar `copula' was briefly investigated in \cite{Tajar01}.

\ppn Following \citet[Section 4]{Fienberg70b}, all $\pp \in \Ps_{2 \times 2}$ with the same marginal distributions must lie on a straight line orthogonal to the edges $\dd_1 \dd_4$ and $\dd_2 \dd_3$ in the Bernoulli tetrahedron. Denote 
\begin{equation}  \overline{\mm} = \begin{pmatrix} 1/2 & 0 \\ 0 & 1/2 \end{pmatrix} \qquad \text{ and } \qquad \overline{\ww} = \begin{pmatrix} 0 & 1/2 \\ 1/2 & 0 \end{pmatrix}, \label{eqn:FrechBerncoppmf} \end{equation}
the mid-points of $\dd_1 \dd_4$ and $\dd_2 \dd_3$, respectively. The segment $\overline{\ww} \,\overline{\mm}$, shown as a thick line in Figure \ref{fig:tetra}, is orthogonal to both $\dd_1 \dd_4$ and $\dd_2 \dd_3$. Hence it consists of all those distributions which share the same margins as $\overline{\mm}$ and $\overline{\ww}$, which are Bernoulli$(1/2)$ for both $X$ and $Y$, that is, all the Bernoulli copula pmf's. The element (\ref{eqn:Berncoppmf}) is the {\it unique} intersection between $\overline{\ww} \,\overline{\mm}$ and the nucleus $[\pp]$ characterised by the odds ratio $\omega(\pp) = \omega$.

\ppn In particular, given that $X \indep Y \iff \omega = 1$, the independence Bernoulli copula pmf is evidently
\begin{equation*} \overline{\ppi} \doteq \begin{pmatrix} 1/4 & 1/4 \\ 1/4 & 1/4 \end{pmatrix}, \label{eqn:indBerncop} \end{equation*}
as one could expect, and clearly $X \indep Y \iff \overline{\pp} = \overline{\ppi}$. This can be contrasted to the observation made in Section \ref{sec:copfordiscr} that a continuous copula $C$ gluing two independent Bernoulli's as in (\ref{eqn:Sklar}) need not be the independence copula. Note that $\overline{\ppi}$ is the centre of gravity of the Bernoulli tetrahedron (Figure \ref{fig:tetra}). It is also the neutral element for the `perturbation' operator (Remark \ref{rmk:cda}): $\forall \pp \in \Ps_{2\times 2}$, $\pp \oplus \overline{\ppi} = \pp = \overline{\ppi} \oplus \pp$.


\begin{example} We call `{\it confetti plot}' the below -- naive but effective -- visual display of bivariate Bernoulli pmf's and their copulas. The size and the colour of the dots are proportional to the value of the corresponding probability. Marginal probabilities are shown as black dots on the same scale. Figure \ref{fig:coppmf} shows such plots for \cite{Yule12}'s comparison of three hospitals on vaccination and recovery for smallpox patients. The top row shows the initial bivariate Bernoulli distributions as tables like (\ref{eqn:bivbernpmf}) (left: Sheffield; middle: Leicester; right: Homerton and Fulham) -- details and exact figures to be found in \citet[Tables I, III, IV]{Yule12}. The nature of the dependence an how it compares across those three distributions is not obvious visually, owing to their different and largely unbalanced margins. The bottom row shows the corresponding Bernoulli copula pmf's, together with the respective values of $\omega$ (and $\Upsilon$, see Section \ref{subsec:Yule2}). Those copula pmf's make it clear, visually, that the dependence is positive and of similar magnitude across the three distributions, although with a slight decrease from Sheffield to Leicester and finally Homerton and Fulham. This is obviously confirmed by the observed decreasing values of $\omega$ from left to right. \qed

\begin{figure}[h]
\centering
\includegraphics[width=0.7\textwidth]{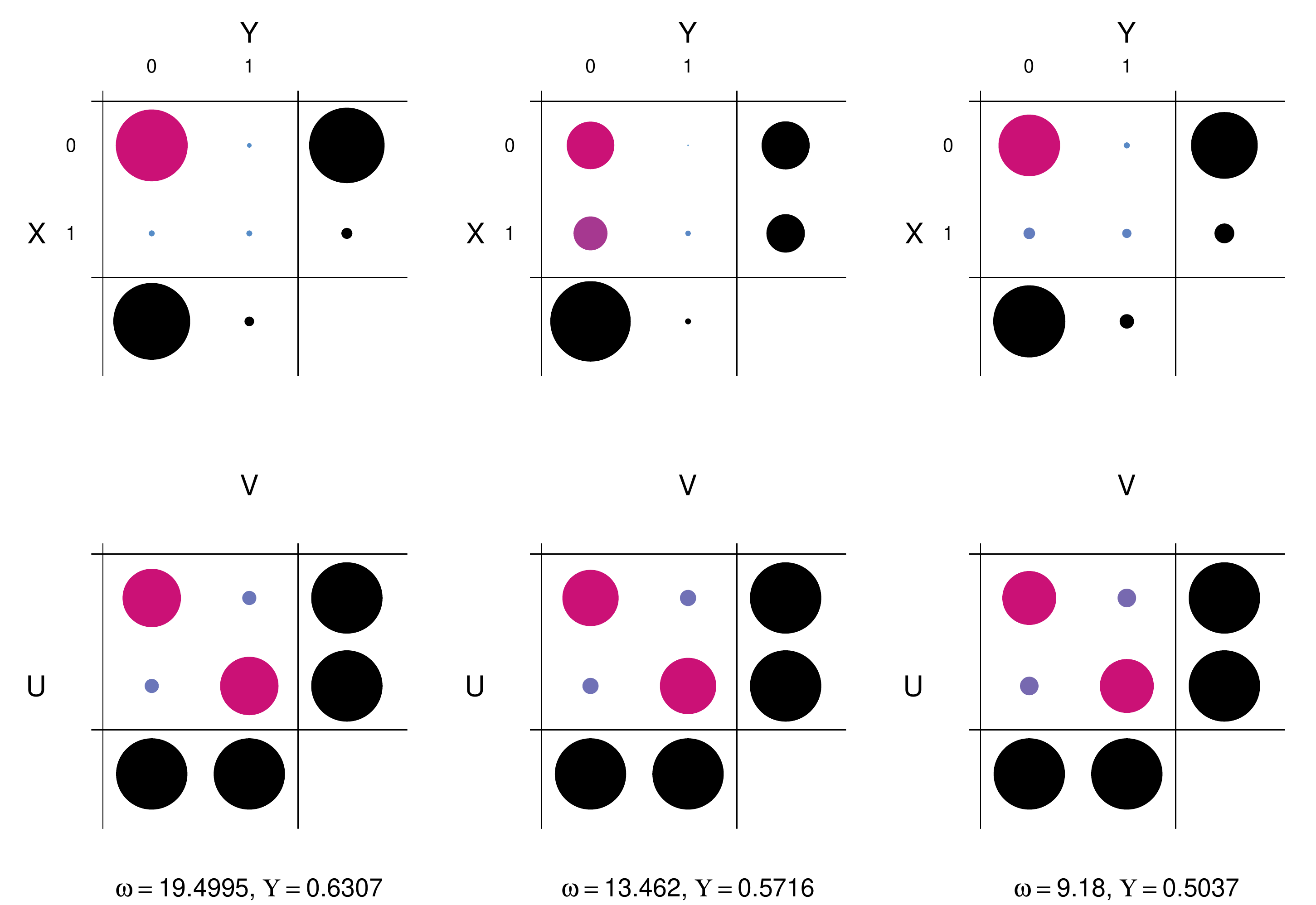}
\caption{Confetti plots of Bernoulli pmf's (top row) and Bernoulli copula pmf's (bottom row) for \cite{Yule12}'s smallpox vaccination data: left: Sheffield; middle: Leicester; right: Homerton and Fulham. }
\label{fig:coppmf}
\end{figure}
\end{example} 

\subsection{Structural zeros} \label{sec:structzeros}

The limit values $\omega = 0$ and $\omega = \infty$ occur when (at least) one of the entries of (\ref{eqn:matBivBern}) is 0, what is usually referred to as a `structural zero' of $\pp$. Such structural zeros are known to create complications in the analysis of comparable contingency tables \citep[Chapter 5]{Bishop75}. Here, with $\omega = 0$ and $\omega = \infty$, we respectively get from (\ref{eqn:Berncoppmf}) the copula pmf's $\overline{\ww}$ and $\overline{\mm}$ defined in (\ref{eqn:FrechBerncoppmf}), with probability mass concentrated on the diagonals. It can be checked that $\overline{\ww}$ and $\overline{\mm}$ are the pmf's associated to the Fr\'echet lower and upper bounds \citep[Sections 2.2 and 2.5]{Nelsen06} in the class of bivariate Bernoulli distributions with Bern$(1/2)$-margins. This is obviously in agreement with the values $\omega = 0$ and $\omega = \infty$ indicating maximal negative and positive dependence between $X$ and $Y$. In the Bernoulli tetrahedron (Figure \ref{fig:tetra}), any Bernoulli copula pmf lies between the `bounds' $\overline{\ww}$ and $\overline{\mm}$ on the segment $\overline{\ww} \,\overline{\mm}$, indeed.

\ppn The value $\omega =0$ arises from distributions like
\begin{equation} (i) \quad \pp_1 = \begin{pmatrix} 0 & \times \\ \times & 0 \end{pmatrix}; \quad \text{ or } \quad   (ii) \quad \pp_2 = \begin{pmatrix} 0 & \times \\ \times & \times \end{pmatrix}, \pp_3 = \begin{pmatrix} \times & \times \\ \times & 0 \end{pmatrix} , \label{eqn:Bernlim} \end{equation}
where $\times$'s represent non-zero elements. In the terminology of \citet[Sections 33.6-33.7]{Kendall61}, case $(i)$ corresponds to {\it absolute} association, whereas case $(ii)$ corresponds to {\it complete} association. Clearly $(i)$ represents `perfect dependence' (negative, in this case), but it is not that clear for $(ii)$ as there is no one-to-one correspondence between rows and columns. Therefore, the fact that the odds-ratio $\omega = 0$ and the corresponding copula pmf $\overline{\ww}$ do not distinguish between $(i)$ and $(ii)$ may appear puzzling. Yet it is actually sensible. Indeed, it is easily seen that `absolute association' $(i)$ is only possible if $\pi_X+\pi_Y = 1$. Whenever $\pi_X + \pi_Y \neq 1$, any sense of `perfect dependence' automatically translates into `complete association'. Hence the dependence is actually as strong as can be in both cases $(i)$ and $(ii)$ {\it given the margins}. A marginal feature, the distinction between $(i)$ and $(ii)$ must be ignored by the copula pmf.

\ppn In the Bernoulli tetrahedron, distributions showing `absolute association' ($\pp_1$) lie on the edge $\dd_2 \dd_3$, while `complete association' characterises the two adjacent (open) faces $\dd_2 \dd_3 \dd_4$ (distributions of type $\pp_2$) and $\dd_1 \dd_2 \dd_3$ (distributions of type $\pp_3$). Actually, the union $\large(\dd_2 \dd_3 \cup \dd_2 \dd_3 \dd_4 \cup \dd_1 \dd_2 \dd_3\large)$ forms the limit of the surfaces of constant odds ratio $\omega$ as $\omega \to 0$. However, there is no transformation (\ref{eqn:Berntrans}) making $\pp_2$ or $\pp_3$ into $\pp_1$, or even $\pp_2$ into $\pp_3$, as a new zero would need to be created, which is not possible given that $\phi$ and $\psi$ are positive in (\ref{eqn:Berntrans}). Clearly, $[\pp_1] = \dd_2 \dd_3$, $[\pp_2] =\dd_2 \dd_3 \dd_4 $ and $[\pp_3] = \dd_1 \dd_2 \dd_3$ are three distinct elements of $\overline{\Ps}_{2\times 2}$, that is, three distinct Bernoulli nuclei, although $\omega(\pp_1) = \omega(\pp_2) = \omega(\pp_3)=0$.

\ppn The edge $[\pp_1] = \dd_2 \dd_3$ physically intersects with $\overline{\ww} \,\overline{\mm}$ (at $\overline{\ww}$, obviously), and indeed, the transformation making $\pp_1$ into $\overline{\ww}$ is trivial and can be written under the form (\ref{eqn:Berntrans}), thus $\overline{\ww} \in [\pp_1]$. By contrast, the (open) faces $[\pp_2] = \dd_2 \dd_3 \dd_4$ and $[\pp_3] = \dd_1 \dd_2 \dd_3$ do not intersect with $\overline{\ww} \,\overline{\mm}$, thus $\overline{\ww} \notin [\pp_2]$ and $\overline{\ww} \notin [\pp_3]$. What is true, though, is that $\overline{\ww}$ is a limit point of $[\pp_2]$ and $[\pp_3]$, as one can approach $\overline{\ww}$ arbitrarily close while staying on any of the two faces. In terms of (\ref{eqn:Berntrans}), for $r \in \{2,3\}$, there exist some sequences of matrices $\Phi^{(r)}_1,\Phi^{(r)}_2,\ldots$ and $\Psi^{(r)}_1,\Psi^{(r)}_2,\ldots$ all in $\Ds_{2\times 2}^{(1)}$ such that, 
\begin{equation*}
\overline{\ww} = \frac{\left(\prod_{k=1}^\infty \Phi^{(r)}_k \right) \cdot \pp_r \cdot \left( \prod_{k=1}^\infty \Psi^{(r)}_k \right)}{\|\left(\prod_{k=1}^\infty \Phi^{(r)}_k \right) \cdot \pp_r \cdot \left( \prod_{k=1}^\infty \Psi^{(r)}_k \right)\|_1}. \label{eqn:limp}
\end{equation*}
In these critical cases $(ii)$, it is thus necessary to extend the nuclei $[\pp_2]$ or $[\pp_3]$ to their closure for them to include the corresponding copula pmf $\overline{\ww}$. Further characterisation of such cases will be given in more generality in Section \ref{sec:RScop}. The case $\omega = \infty$ and $\overline{\pp} = \overline{\mm}$ is treated in perfect analogy.


\begin{remark} \label{rmk:orbs} It is seen that there is no one-to-one correspondence between $[\pp]$ and $\omega({\pp})$. What differentiates $[\pp_1]$, $[\pp_2]$ and $[\pp_3]$ above is the different layout of structural zeros of their elements, that is, their support $\Supp(\pp)= \{(x,y) \in \{0,1\} \times \{0,1\}: p_{xy} > 0 \}$. In fact, any nucleus $[\pp]$ is unequivocally characterised by the couple $(\Supp(\pp),\omega(\pp))$. 
This double characterisation clarifies what really makes the dependence structure of a bivariate Bernoulli vector $(X,Y)$. First, the dominant effect is the presence of structural zeros and their layout: in the presence of one (or two) such zero(s), the value of the other non-null $p_{xy}$'s is irrelevant. In fact, any structural zero in $\pp$ implies, by definition, the incompatibility of two particular values taken by $X$ and $Y$, ergo it is a critical constituent of the dependence between $X$ and $Y$. When all $p_{xy}$ are positive, then in some sense the dependence is more subtle, and fully characterised by $\omega$. \qed
\end{remark}


\subsection{Yule's colligation coefficient} \label{subsec:Yule2}

\ppn Suppose that $(U,V)$ is a discrete random vector supported on $\{\frac{1}{3},\frac{2}{3}\} \times \{\frac{1}{3},\frac{2}{3}\}$ with joint pmf as in (\ref{eqn:Berncoppmf}) for some $\omega \geq 0$. One can check that Pearson's correlation between $U$ and $V$ is
\begin{equation} \Upsilon = \frac{\sqrt{\omega}-1}{\sqrt{\omega}+1}, \label{eqn:YuleY} \end{equation}
which is exactly Yule's `{\it colligation coefficient} $Y$' \citep[pp.\,592-593]{Yule12}. Hence, Yule's $Y$, denoted here Upsilon $\Upsilon$, can be regarded as the `Bernoulli analogue' to Spearman's rank correlation $\rho$ in the continuous case in the sense that it is the Pearson's correlation of the distribution of interest {\it after copula transformation}. Evidently $\Upsilon$ is margin-free, as it is a one-to-one function of $\omega$. In addition, as Pearson's correlation is invariant under linear transformations, $\Upsilon$ remains unaffected as well if re-defining the margins of (\ref{eqn:Berncoppmf}) as $U^* = a_1 U + b_1$ and $V^* = a_2 V + b_2$ (with $a_1 a_2 > 0$). Again, the exact `labels' of the rows and columns of (\ref{eqn:Berncoppmf}) do not play any role (provided their relative order is preserved). 

\ppn Obviously, $\Upsilon = 0 \iff \omega =1 \iff X \indep Y$ in (\ref{eqn:bivbernpmf}). In addition, $|\Upsilon|$ does attain its maximum value 1 when $\omega=0$ or $\omega = \infty$, which corresponds to the Fr\'echet bounds described in Section \ref{sec:structzeros}. These two observations can be contrasted to the inconsistencies with the usual copula-based Spearman's $\rho$ \citep[Sections 4.2 and 4.4]{Genest07b}. Actually, it is known \citep{Cureton59} that Pearson's correlation computed on a binary contingency table can only reach the values $-1$ and $+1$ when $ p_{1 \bullet} = p_{\bullet 1}  = 0.5$. This is what is achieved by the Bernoulli copula transformation.

\ppn Kendall's $\tau$ corrected for the occurrence of ties (`$\tau_b$' in \citet{Genest07b}) is given, for the bivariate Bernoulli case, by $\tau_b = \frac{p_{00} - p_{0\bullet}p_{\bullet 0}}{\sqrt{p_{0\bullet}p_{\bullet 0} p_{1 \bullet}p_{\bullet 1}}}$. This, computed on the copula pmf (\ref{eqn:Berncoppmf}), reduced down to $\tau_b = \Upsilon$ again. This yields the same conclusion as above about the maximum values reached by $\tau_b$, which again can be contrasted with \citet[Section 4.6]{Genest07b}. 

\ppn In fact, it can be checked that many other classical association measures for contingency tables reduce down to $\Upsilon$ or $|\Upsilon|$ when computed on the Bernoulli copula distribution (\ref{eqn:Berncoppmf}), and those include Cram\'er's $V$ \citep[Chapter 21]{Cramer46}, Goodman and Kruskal's $\lambda$ \citep[Section 5.2]{Goodman54} or Cohen's $\kappa$ \citep{Cohen60}. This suggests that Yule's $\Upsilon$ is a very natural, if not the canonical, dependence parameter in the $(2 \times 2)$-table framework. 
Indeed, given that  (\ref{eqn:YuleY}) can be reversed as $\sqrt{\omega} = (1+\Upsilon)/(1-\Upsilon)$, the copula pmf $\overline{\pp}$ (\ref{eqn:Berncoppmf}) can be written under the even simpler form
\begin{equation} \overline{\pp}= \frac{1}{4} \begin{pmatrix} 1+\Upsilon & 1-\Upsilon \\ 1-\Upsilon & 1+\Upsilon \end{pmatrix} \label{eqn:BerncoppmfYule} \end{equation}
in terms of $\Upsilon \in [-1,1]$. The effect of $\Upsilon$ on $\overline{\pp}$ is thus linear in nature. The value of $\Upsilon$ acts as a ruler along $\overline{\ww} \, \overline{\mm}$ in Figure \ref{fig:tetra}: from $\Upsilon = -1$ at $\overline{\ww}$ to $\Upsilon = 1$ at $\overline{\mm}$, via $\Upsilon = 0$ at $\overline{\ppi}$.

\subsection{Construction of arbitrary bivariate Bernoulli distributions with given copula pmf} \label{sec:arbbern}

The Bernoulli nucleus $[\pp]$ contains all bivariate Bernoulli distributions given by (\ref{eqn:Berntrans}). \citet[p.\ 375]{Bishop75} noted that ``{\it We can use }[that]{\it \ transformation to change the given marginal distributions into any other set of marginal distributions}''. Indeed, we can construct a bivariate Bernoulli distribution $\pp$ whose marginals are Bern$(\pi_X)$ and Bern$(\pi_Y)$ ($0<\pi_X,\pi_Y<1$) and dependence structure prescribed by a certain Bernoulli copula characterised by (\ref{eqn:Berncoppmf}). For $0 < \omega < \infty$, that $\pp$ is the unique intersection between the set of all bivariate Bernoulli's with the requested margins, which is a segment parallel to $\overline{\ww}\overline{\mm}$ in Figure \ref{fig:tetra} \citep[Section 4]{Fienberg70}, and the surface representing the unique nucleus $[\pp]$ characterised by $\omega(\pp) = \omega$. That element is obtained explicitly by solving the quadratic system
\begin{equation*} \left\{\begin{array}{c c c} p_{10} + p_{11} & = & \pi_X \\ 
p_{01} + p_{11} & = & \pi_Y \\
\frac{p_{00}p_{11}}{p_{10}p_{01}} & = & \omega   \\
p_{00} + p_{01} + p_{10} + p_{11} & = & 1
\end{array} \right. , \label{eqn:Bernsyst}\end{equation*}
which yields, for $\omega \neq 1$,
\begin{equation} p_{11} = \frac{1}{2(\omega-1)} \left\{1+(\omega-1)(\pi_X+\pi_Y) - \sqrt{[1+(\omega-1)(\pi_X+\pi_Y)]^2-4\omega(\omega-1)\pi_X\pi_Y} \right\}, \label{eqn:p11om} \end{equation}
which is formula (2*) in \cite{Mosteller68}. If $\omega = 1$, then trivially $p_{11} =  \pi_X \pi_Y$. The other values follow by substitution. In particular, $p_{00} = 1-\pi_X-\pi_Y+p_{11}$. Of course $p_{11}$ in (\ref{eqn:p11om}) is an increasing function of $\omega$, and hence so is $p_{00}$. In the same time, `the value $C(1-\pi_X,1-\pi_Y)$ singled out by Sklar's theorem', described below (\ref{eqn:OR}), is precisely $p_{00}$, establishing the one-to-one correspondence between $C(1-\pi_X,1-\pi_Y)$ and $\omega$. Describing the dependence in a bivariate Bernoulli distribution by $\omega$, or any monotonic function of $\omega$ such as $\Upsilon$, is thus totally consistent with Sklar's theorem.


\ppn For $\omega = 0$, we must have either $p_{00} = 0$ or $p_{11} =0$ (or both). By obvious substitution, one gets
\begin{equation} \pp_1 = \begin{pmatrix} 0 & 1-\pi_X \\ 1-\pi_Y & 0 \end{pmatrix} , \quad  \pp_2 = \begin{pmatrix} 0 & 1-\pi_X \\ 1-\pi_Y & \pi_X+\pi_Y-1 \end{pmatrix} \quad  \text{ or } \quad \pp_3 = \begin{pmatrix} 1-\pi_X-\pi_Y & \pi_Y \\ \pi_X & 0 \end{pmatrix} , \label{eqn:decompmat} \end{equation}
if $\pi_X+\pi_Y =1$, $\pi_X+\pi_Y > 1$ or $\pi_X +\pi_Y < 1$, respectively. 

\ppn Now, if cases of `absolute association' $\pp_1$ can trivially be reconstructed from $\overline{\ww}$ through a transformation like (\ref{eqn:Berntrans}), it is not true for cases of `complete association' $\pp_2$ or $\pp_3$ given that (\ref{eqn:Berntrans}) does not allow to `escape' from $[\overline{\ww}]=[\pp_1]$ and $\pp_2,\pp_3 \notin [\overline{\ww}]$ (Section \ref{sec:structzeros}). However, the presence of a structural zero makes it actually easier to reconstruct $\pp_2$ or $\pp_3$ as it is enough to adjust the requested marginals around that 0. 

\begin{remark} \label{rmk:unik} In the Bernoulli tetrahedron, this amounts to finding the intersection between the `segment of requested margins', and either the face $\dd_2 \dd_3 \dd_4$ or the face $\dd_1 \dd_2 \dd_3$. In addition, because that segment is parallel to $\overline{\ww} \overline{\mm}$, it can only pierce one of those two faces, establishing geometrically the uniqueness of a distribution $\pp$ with the requested margins and $\omega = 0$. \qed
\end{remark}

The case $\omega=\infty$ is treated identically, by symmetry. Hence any two Bernoulli distributions can be glued together through the Bernoulli copula for showing any level of dependence set by $\omega \in [0,\infty]$.

\begin{example} Consider Table 1 in \cite{Lin09}, which arises from 69 medical malpractice claims. Two surgeon-reviewers were asked to determine whether a communication breakdown occurred during a hand-off in care (with possible answer `Yes' or `No'). It turns out that eight reviews were missing for Surgeon 1 and 11 reviews were missing for Surgeon 2. The raw data are:
\begin{equation} \begin{array}{c c ||c c c c c | c}
& \bsfrac{\text{Surg.}2}{\text{Surg.}1} & & \text{Yes} & \text{No} & \text{Missing} & \\
\hline\hline
& \text{Yes} & & 26 & 1 & 2 & & 29 \\ 
& \text{No} & & 5  & 18 & 9 & & 32  \\
& \text{Missing} & & 4  & 4 & 0 & & 8  \\
\hline 
& & & 35 & 23 & 11 &  & 69 \\
\end{array}. \label{eqn:rawdatLin09} \end{equation}
Focusing only on the 50 complete cases (Yes/No rows and columns only) and identifying `Yes' $=0$ and `No' $=1$ (although this is irrelevant), the corresponding bivariate Bernoulli (empirical) distribution is:
\begin{equation} \begin{array}{c l ||c c c c | c}
& \bsfrac{Y}{X} & & 0 & 1 & \\
\hline\hline
& 0 & & 0.52 & 0.02& & 0.54 \\
& 1 & & 0.10  & 0.36 & & 0.46  \\
\hline 
& & & 0.62 & 0.38 & & 1 \\
\end{array}. \label{eqn:compcasesLin} \end{equation}
It is found that $\omega = 93.6$, $\Upsilon = 0.813$, quantifying the high level of agreement between the two surgeons. The associated Bernoulli copula pmf (\ref{eqn:Berncoppmf})/(\ref{eqn:BerncoppmfYule}) is 
\begin{equation} \overline{\pp} = \begin{pmatrix} 0.453 & 0.047 \\ 0.047 & 0.453 \end{pmatrix}. \label{eqn:copLin09} \end{equation}
Now, in order to use the partially observed cases as well, \cite{Altham10} suggested to analyse (\ref{eqn:rawdatLin09}) as a $(3 \times 3)$-table, with $p_{22}$ being a structural zero given that we cannot observed such Missing-Missing cases. Given that we are rather in a situation of truncation here, it seems wiser to incorporate at best the missing cases to the main table by updating the marginals. By doing so, it seems sensible to assume that the degree of association would be in principle the same between the two surgeons' answers in cases in which one answer is missing as it is in fully observed cases. This is akin to a `Missing-at-Random' assumption. In other words, we would like to produce a bivariate Bernoulli distribution with marginals $(29,32)/61 = (0.475,0.525)$ and $(35,23)/58 = (0.603,0.397)$ (from (\ref{eqn:rawdatLin09})) and copula pmf (\ref{eqn:copLin09}). From (\ref{eqn:p11om}), it is easily found to be:
\begin{equation} \begin{array}{c l ||c c c c | c}
& \bsfrac{\text{Surg.}2}{\text{Surg.}1} & & \text{Yes} & \text{No} & \\
\hline\hline
& \text{Yes} & & 0.462 & 0.013 & & 0.475 \\
& \text{No} & & 0.141  & 0.383 & & 0.525  \\
\hline 
& & & 0.603 & 0.397 & & 1 \\
\end{array}. \label{eqn:corrcasesLin09} \end{equation}
This distribution seems to encompass the whole available information. 
\qed

 \end{example}

\ppn Indeed the bivariate Bernoulli case is rather simple, and most of the essential elements of the Bernoulli copula pmf have been investigated before. For instance, the form (\ref{eqn:Berncoppmf}) or the `quadratic formula' (\ref{eqn:p11om}) can be found in the previous literature. However, we agree with \cite{Faugeras17} that the virtue of the bivariate Bernoulli case is precisely its simplicity, which makes transparent the ideas and issues involved. The extension of those ideas to general $(R \times S)$-discrete distributions is investigated in the next section.

\section{General discrete distributions with finite support} \label{sec:RScop}

Let $(X,Y)$ be a bivariate discrete vector where $X$ and $Y$ may only take a finite number of values. Without loss of generality, let $X \in \Ss_X \doteq \{0,1,\ldots,R-1\}$ and $Y \in \Ss_Y \doteq \{0,1,\ldots,S-1\}$, with $R,S \in \N$, $2 \leq R,S < \infty$. Let $\pp$ be its joint probability mass function, defined by $p_{xy} = \P(X=x,Y=y)$, $(x,y) \in \Ss_X \times \Ss_Y$, and $\pp_X = (p_{0\bullet},p_{1\bullet},\ldots,p_{R-1 \bullet})$ and $\pp_Y = (p_{\bullet 0},p_{ \bullet 1},\ldots,p_{ \bullet S-1})$ its marginal distributions: $p_{x\bullet} = \sum_{y \in \Ss_Y} p_{xy} = \P(X=x)$ and $p_{\bullet y}= \sum_{x \in \Ss_X} p_{xy}  = \P(Y=y)$. Let $\Ps_{R \times S}$ be the set of all such bivariate discrete distributions $\pp$ with $p_{x \bullet} >0 \ \forall x\in \Ss_X$ and $p_{\bullet y} >0 \ \forall y\in \Ss_Y$, identified to the $(R \times S)$-matrices

\begin{equation*} \pp = \begin{pmatrix}
p_{00} & p_{01} &   \ldots & p_{0,S-1} \\
p_{10} & p_{11} &   \ldots & p_{1,S-1} \\
\vdots & \vdots &  \ddots  & \vdots \\
p_{R-1,0} & p_{R-1,1} &   \ldots & p_{R-1,S-1} 
\end{pmatrix}.  \label{eqn:ppRS} \end{equation*}

Any such distribution can be regarded as a point in the $(RS-1)$-dimensional simplex. Most of the ideas described below have, therefore, a geometric interpretation similar to Figure \ref{fig:tetra}, see \cite{Fienberg68}.

\subsection{Odds ratio matrix}

Sklar's Theorem establishes that one must be able to entirely describe the inner dependence structure of such a $(R \times S)$-bivariate discrete distribution by $(R-1)(S-1)$ parameters, as here $\Ran F_X \times \Ran F_Y$ consists of $(R-1)(S-1)$ informative locations (i.e., strictly inside the unit square). This is consistent with the usual break down of degrees of freedom in comparable $(R \times S)$-contingency tables: from $RS-1$ for an unconstrained table, minus $(R-1)$ when one fixes the row marginal distribution and $(S-1)$ when one fixes the column marginal distribution, so that there remain $(R-1)(S-1)$ degrees of freedom for describing the association structure of the table. Those $(R-1)(S-1)$ parameters can be a family of odds-ratios (Altham, \citeyear{Altham70}; Agresti, \citeyear[Section 2.4.1]{Agresti13}) -- at least when there is no structural zero in the table. \citet[p.\ 55]{Agresti13} stressed that `{\it given the marginals, the odds ratios determine the cell probabilities}'; in other words, the full distribution can be entirely reconstructed by {\it coupling} the marginal distributions and the set of odds ratios. Again, those entirely fulfil here the desired role of classical copulas, the explicit resort to which being consequently purposeless. 

\ppn Let 
\begin{equation} \omega_{xy} = \frac{p_{00}p_{xy}}{p_{0y}p_{x0}}, \quad \forall (x,y) \in \Ss_X\backslash \{0\} \times \Ss_Y\backslash \{0\} \label{eqn:omxy} \end{equation}
be the odds ratio of the bivariate Bernoulli pmf
\begin{equation*} \pp_{xy} = \frac{1}{K_{xy}}\begin{pmatrix} 
p_{00} & p_{0y} \\
p_{x0} & p_{xy}
\end{pmatrix}, \label{eqn:subtabBern} \end{equation*}
where $K_{xy} = \P(X \in \{0,x\}, Y \in  \{0,y\})$ is a normalisation constant irrelevant in (\ref{eqn:omxy}). 

\ppn Call $\Ms^{(+)}_{(R-1)\times (S-1)}$ the set of all $(R-1)\times(S-1)$-matrices with non-negative, possibly infinite, entries. Define the map
\begin{equation} \Omega: \Ps_{R\times S} \to \Ms^{(+)}_{(R-1)\times (S-1)}: \pp \to \Omega(\pp) = [\omega_{xy}]_{\substack{x = 1,\ldots,R-1, \\ y = 1,\ldots,S-1}} \label{eqn:ORmat} \end{equation}
where $\omega_{xy}$ is given by (\ref{eqn:omxy}). Thus, $\Omega(\pp)$, called the {\it odds ratio matrix}, is the matrix whose element $(x,y)$ is $\omega_{xy}$, $(x,y) \in \Ss_X\backslash \{0\} \times \Ss_Y\backslash \{0\}$. 

\begin{remark} \label{rmk:undef} Although they were ruled out in  (\ref{eqn:bivbernpmf})-(\ref{eqn:matBivBern}) when assuming $0< \pi_X,\pi_Y < 1$, cases of $0/0$ may arise in (\ref{eqn:omxy}). Then the corresponding entry of $\Omega(\pp)$ may be left undefined. Admitting some slight lack of rigour, we identify two odds ratio matrices whose all well defined entries are equal, i.e., an undefined entry in a matrix is assumed to be equal to whatever the corresponding entry may be in the other. \qed
\end{remark}

\subsection{Marginal transformations and nucleus} \label{sec:margtransnuc}

Define $\Ds^{(1)}_{Q \times Q}$ the set of all diagonal $Q \times Q$ matrices whose entry $(1,1)$ is equal to 1 and other diagonal entries are positive. Similarly to (\ref{eqn:Berntrans}), for any $\Phi \in \Ds^{(1)}_{R \times R}$ and $\Psi \in \Ds^{(1)}_{S \times S}$, let
\begin{equation} g_{\Phi,\Psi}: \Ps_{R\times S} \to \Ps_{R \times S}: g_{\Phi,\Psi}(\pp) =  \frac{ \Phi \cdot \pp \cdot \Psi}{\| \Phi \cdot \pp \cdot \Psi\|_1}. \label{eqn:gphipsi} \end{equation}
The matrix $\Phi$ multiplies the rows of $\pp$ and the matrix $\Psi$ multiplies the columns of $\pp$: this is akin to `marginal distortions' as in Section \ref{subsec:margtrans}. This defines a group action on $\Ps_{R \times S}$, which induces orbits $[\pp]$: \begin{equation} \pp \sim \pp^* \iff [\pp] = [\pp^*] \iff \exists (\Phi ,\Psi) \in  \Ds^{(1)}_{R \times R} \times \Ds^{(1)}_{S \times S} \text{ s.t. } \pp^* = g_{\Phi,\Psi}(\pp). \label{eqn:orbit} \end{equation}
Further, define a limit point of $[\pp]$ as an element of $\Ps_{R\times S}$ which can be written as
\begin{equation}
\frac{\left(\prod_{k=1}^\infty \Phi_k \right) \cdot \pp \cdot \left( \prod_{k=1}^\infty \Psi_k \right)}{\|\left(\prod_{k=1}^\infty \Phi_k \right)\cdot \pp \cdot \left( \prod_{k=1}^\infty \Psi_k \right)\|_1} \label{eqn:limg}
\end{equation}
for some sequences of matrices $\Phi_1,\Phi_2,\ldots \in \Ds_{R\times R}^{(1)}$ and $\Psi_1,\Psi_2,\ldots \in \Ds_{S\times S}^{(1)}$. Let $\Cl([\pp])$ be the closure of $[\pp]$, that is, the union of $[\pp]$ and its limit points.

\ppn Free from any sense of marginal distributions, the orbits $[\pp]$ must again be equivalence classes of dependence. Indeed, for any two $\pp, \pp^*\in \Ps_{R \times S}$,
\begin{equation}  \pp \sim \pp^*  \Rightarrow \Omega(\pp) = \Omega(\pp^*),  \label{eqn:equiclassdep} \end{equation}
as easily follows from the fact that all odds ratios are preserved by $g_{\Phi,\Psi}$, exactly as in (\ref{eqn:comOR}). This holds true for any undefined elements of $\Omega(\pp)$ as well, as $g_{\Phi,\Psi}$ leaves the zeros of $\pp$ unaffected. Hence $[\pp]$ will again be called the nucleus of the discrete pmf $\pp$, as it characterises the inner dependence structure of $(X,Y)$. 
Note that, if all entries of $\Omega(\pp)$ are defined and positive, then $\Omega(\pp) = \Omega(\pp^*) \Rightarrow [\pp]= [\pp^*]$. Like in Remark \ref{rmk:orbs}, though, one may find two $\pp_1, \pp_2 \in \Ps_{R \times S}$ with $\Omega(\pp_1) = \Omega(\pp_2)$ but $[\pp_1] \neq [\pp_2]$ when $\Supp(\pp_1) \neq \Supp(\pp_2)$, that is, when $\pp_1$ and $\pp_2$ show a different pattern of structural zeros. Again, the preponderant role of structural zeros on the dependence structure appears clearly. 

\ppn Like in Section \ref{subsec:Berncoppmf}, one may wish to single out the member of $[\pp]$ with uniform marginals for embodying the dependence pattern in $\pp$ by a simple element of the class. That one would be called the `copula pmf' of $\pp$, leading to the following definition of a discrete copula, the obvious analogue to Definition \ref{dfn:classcop}.

\begin{definition} \label{dfn:discrcop} A bivariate $(R \times S)$-discrete copula is the bivariate $(R \times S)$-discrete distribution of a vector $(U,V)$ whose both marginal distributions are discrete uniform on $\Ss_U \doteq \{\frac{1}{R+1},\frac{2}{R+1},\ldots,\frac{R}{R+1}\}$ and $\Ss_V \doteq \{\frac{1}{S+1},\frac{2}{S+1},\ldots,\frac{S}{S+1}\}$, respectively. The associated copula probability mass function (copula pmf) is thus a bivariate discrete pmf $\overline{\pp}$ on $\Ss_U \times \Ss_V$ such that for all $u \in \{0,\ldots,R-1\}$, $\sum_{v=0}^{S-1} \overline{p}_{uv} = \frac{1}{R}$, and for all $v \in \{0,\ldots,S-1\}$, $\sum_{u=0}^{R-1} \overline{p}_{uv} = \frac{1}{S}$, where $\overline{p}_{uv} = \P(U=\frac{1}{R+1}+\frac{u}{R+1},V=\frac{1}{S+1}+\frac{u}{S+1})$.
\end{definition} 

\begin{remark} A very similar definition is given in \cite{Kolesarova06}, who investigated such `discrete copulas' in the case $R=S$. The `copula pmf' here coincides essentially with the (rescaled) bistochastic matrix of their Proposition 2. See also \cite{Mayor05,Mayor07,Aguilo06,Kobayashi14} or \cite{Deamo17}. Those papers investigate the analytical properties of such matrices, though, and show little overlap with what is discussed here. See also \citet[Section 3.1.1]{Durante16}.\qed 
\end{remark}

\ppn Naturally, defining the copula pmf of $\pp$ as the member of $[\pp]$ with uniform margins raises the question of the existence and uniqueness of such an element on $[\pp]$.


\subsection{Existence and uniqueness of the copula pmf}

This question is actually heavily linked to the problem of `matrix scaling': `{\it Given a nonnegative matrix $A$, can we find diagonal matrices $D_1$ and $D_2$ such that $D_1 A D_2$ is doubly stochastic?}' \citet{Sinkhorn64} showed that the answer is affirmative if $A$ is a positive square matrix. Later, this result was generalised in many directions, including to non-negative and/or non-square matrices; 
see \cite{Idel16} for a review. 

\ppn From recent results in the field, a simple necessary and sufficient criterion for the existence and uniqueness of the copula pmf of a given $\pp \in \Ps_{R \times S}$ can be formulated. When all $p_{xy}$'s are positive in $\pp$, it can directly be deduced from \citet{Sinkhorn64,Sinkhorn67} that the copula pmf exists and is unique. Hence the defining criterion will be again the presence of structural zeros in $\pp$ and their layout. Define $\Supp(\pp) = \{(x,y)\in \Ss_X \times \Ss_Y \text{ s.t. } p_{xy} >0 \}$ the support of $\pp$, and $N(\pp) = \left\{(\nu_X \times \nu_Y): \nu_X \subset \Ss_X, \nu_Y \subset \Ss_Y \ \text{s.t.}\ \sum_{(x,y) \in \nu_X \times \nu_Y} p_{xy} =0 \right\}$, the set of {\it rectangular} subsets of $\Ss_X \times \Ss_Y$ on which $\pp$ is null. Naturally, $N(\pp) = \emptyset$ if and only if $\Supp(\pp) = \Ss_X \times \Ss_Y$, i.e.\ $p_{xy} >0 $ for all $(x,y)$.



\ppn Let $\Cs_{R \times S} = \{\pp \in \Ps_{R \times S}: p_{x\bullet} = \frac{1}{R} \ \forall x\in \Ss_X, p_{\bullet y} = \frac{1}{S} \ \forall y\in \Ss_Y \} \subset \Ps_{R \times S}$, the set of all $(R \times S)$-discrete copulas as per Definition \ref{dfn:discrcop}. For any discrete set $A$, denote $|A|$ the number of elements in $A$. Then: 

\begin{theorem} \label{thm:coppmf} Let $\pp \in \Ps_{R \times S}$. 
\begin{enumerate}[(a)]
\item Suppose that, for all $(\nu_X \times \nu_Y) \in N(\pp)$, $\frac{|\nu_X|}{R} + \frac{|\nu_Y|}{S} < 1$. Then, there exists a unique $\overline{\pp} \in [\pp] \cap \Cs_{R \times S}$;
\item Suppose that, for all $(\nu_X \times \nu_Y) \in N(\pp)$, $\frac{|\nu_X|}{R} + \frac{|\nu_Y|}{S} \leq 1$, with $\frac{|\tilde{\nu}_X|}{R} + \frac{|\tilde{\nu}_Y|}{S} = 1$ for some $(\tilde{\nu}_X \times \tilde{\nu}_Y) \in N(\pp)$. 
\begin{enumerate}[(i)] 
\item If, for all $(\tilde{\nu}_X \times \tilde{\nu}_Y) \in N(\pp)$ such that $\frac{|\tilde{\nu}_X|}{R} + \frac{|\tilde{\nu}_Y|}{S} = 1$, $(\Ss_X \backslash \tilde{\nu}_X \times \Ss_Y \backslash \tilde{\nu}_Y) \in N(\pp)$, then there exists a unique $\overline{\pp} \in [\pp] \cap \Cs_{R \times S}$;
\item If there exists $(\tilde{\nu}^*_X \times \tilde{\nu}^*_Y) \in N(\pp)$ such that $\frac{|\tilde{\nu}^*_X|}{R} + \frac{|\tilde{\nu}^*_Y|}{S} = 1$ and $(\Ss_X \backslash \tilde{\nu}^*_X \times \Ss_Y \backslash \tilde{\nu}^*_Y) \notin N(\pp)$, then $[\pp] \cap \Cs_{R \times S} = \emptyset$ but there exists a unique $\overline{\pp} \in \Cl\left([\pp]\right) \cap \Cs_{R \times S}$;
\end{enumerate}
\item Suppose that there exists $(\doublewidetilde{\nu}_X \times \doublewidetilde{\nu}_Y) \in N(\pp)$ such that $\frac{|\doublewidetilde{\nu}_X|}{R} + \frac{|\doublewidetilde{\nu}_Y|}{S} > 1$. Then, $\Cl\left([\pp]\right) \cap \Cs_{R \times S} = \emptyset$.
\end{enumerate}
\end{theorem}
\begin{proof} This is essentially Theorems 3.1, 3.2 and 3.3 in \cite{Brossard17}. \end{proof}

Case $(a)$ is the `easy' case, which obviously covers all $(R \times S)$-distributions with no structural zeros ($N(\pp) = \emptyset$), but not only: $\pp$ is allowed to have structural zeros ($N(\pp) \neq \emptyset$), provided those are not too `prominent' in the specified sense. The unique $\overline{\pp}$ is obviously the copula pmf of $\pp$: it belongs to $\Cs_{R \times S}$ and it is such that $\Omega(\overline{\pp}) = \Omega(\pp)$ (in the sense of Remark \ref{rmk:undef}), as $\overline{\pp} \in [\pp]$ and (\ref{eqn:equiclassdep}). Note that, from (\ref{eqn:orbit}), $\Supp(\overline{\pp}) = \Supp(\pp)$, that is, the pattern of structural zeros (if any) is the same in $\pp$ and $\overline{\pp}$. 

\ppn Case $(b)$ is the critical case. In case $(b\,(i))$, the matrix $\pp$ can be written under a `disconnected' form, that is, it can be made block-diagonal by some permutations of its rows and columns. Then, each sub-block of non-zero elements of $\pp$ can be dealt with separately when adjusting the margins, and it remains possible to write $\overline{\pp}$ under the form (\ref{eqn:gphipsi}), that is $\overline{\pp} \in [\pp]$. This preserves the pattern of zeros as well, and $\Supp(\overline{\pp}) = \Supp(\pp)$. In the Bernoulli case, this corresponds to `absolute association' ($(i)$ in (\ref{eqn:Bernlim})). 

\ppn By contrast, in case $(b\,(ii))$, the matrix $\pp$ is `connected'. For complying with the uniform margins constraint, new zeros must be created in $\overline{\pp}$ which must therefore be a limit point of $[\pp]$ in the sense (\ref{eqn:limg}): $\overline{\pp} \in \Cl([\pp])$ and $\Supp(\overline{\pp}) \subset \Supp(\pp)$. The new zeros are created on $\bigcup_{(\tilde{\nu}^*_X \times \tilde{\nu}^*_Y) \in N(\pp)} (\Ss_X \backslash \tilde{\nu}^*_X \times \Ss_Y \backslash \tilde{\nu}^*_Y)$. But it holds true that $\Omega(\overline{\pp}) = \Omega(\pp)$ (in the sense of Remark \ref{rmk:undef}) and $\overline{\pp} \in \Cs_{R \times S}$, hence $\overline{\pp}$ is again the unique copula pmf of $\pp$. In the Bernoulli case, this corresponds to `complete association' ($(ii)$ in (\ref{eqn:Bernlim})).

\ppn Finally, case $(c)$ establishes the non-existence of a copula pmf when the structural zeros form a bulky subset of $\pp$. As the zeros cannot be turned into positive values by (\ref{eqn:gphipsi}) and are frozen, there do not remain sufficiently many degrees of freedom for adjusting the marginals. Pragmatically, the dependence between $X$ and $Y$ is so overly dictated by the structural zeros that an approach based on odds ratios is pointless.

\ppn The above observations allow us to state:
\begin{corollary}[{\bf Existence and uniqueness of the copula pmf}] \label{cor:coppmf} The bivariate discrete distribution $\pp \in \Ps_{R \times S}$ admits a unique copula pmf $\overline{\pp}$ if and only if $\frac{|\nu_X|}{R} + \frac{|\nu_Y|}{S} \leq 1$ for all $(\nu_X \times \nu_Y) \in N(\pp)$. By definition, the copula pmf $\overline{\pp}$ has discrete uniform margins, is such that $\Omega(\overline{\pp}) = \Omega(\pp)$ (in the sense of Remark \ref{rmk:undef}) and $\Supp(\overline{\pp}) \subseteq \Supp(\pp)$.
\end{corollary}

\subsection{Iterated proportional fitting procedure} \label{sec:IPF}

Corollary \ref{cor:coppmf} establishes the existence and uniqueness of the copula pmf $\overline{\pp}$ of $\pp$ under mild conditions. 
However, unlike in the Bernoulli case (\ref{eqn:Berncoppmf}), $\overline{\pp}$ is usually not available in closed form in the general $(R \times S)$-case.\footnote{Specific models lead to closed form copula pmf, though; see Section \ref{sec:paramcop}.} An iterative construction of $\overline{\pp}$  consists of alternately normalising the rows and columns of $\pp$ to have uniform marginals. That type of procedure, known as {\it iterated proportional fitting} (IPF), has been used in statistics since \cite{Deming40}, and its convergence was investigated in \cite{Ireland68,Fienberg70,Csiszar75} and \cite{Ruschendorf95}; see also \citet[Section 3.6]{Bishop75}. More recent results \citep{Brossard17} guarantee that the IPF seeded on any $\pp \in \Ps_{R \times S}$ indeed converges to its copula pmf $\overline{\pp}$ provided that it exists, i.e.\ under the condition of Corollary \ref{cor:coppmf}. More precisely, the convergence of the IPF procedure is geometric in case $(a)$ and $(b\,(i))$, but slower in case $(b\,(ii))$ \citep[Theorems 3.2 and 3.3]{Brossard17}. The copula pmf $\overline{\pp}$ of any bivariate discrete distribution admitting one can thus be obtained almost instantly. The R package {\tt mipfp} \citep{Barthelemy15} provides an easy implementation of IPF, which was used for the examples in this paper.

\subsection{Construction of arbitrary bivariate discrete distributions with given copula pmf} \label{sec:arbdistrRS}

\ppn Similarly to Section \ref{sec:arbbern}, one may want to construct a $(R \times S)$-discrete distribution with particular marginal distributions $\pp_X$ and $\pp_Y$ and dependence structure driven by a copula pmf $\overline{\pp}$. The question of the existence and uniqueness of such a distribution is (partially) given by the following result, analogue to Theorem \ref{thm:coppmf}.


\begin{theorem} Let $\pp_X = (p_{0\bullet},p_{1\bullet},\ldots,p_{R-1 \bullet})$ and $\pp_Y = (p_{\bullet 0},p_{ \bullet 1},\ldots,p_{ \bullet S-1})$ be some target marginal distributions for $X$ and $Y$, respectively. Let $\overline{\pp} \in \Cs_{R\times S}$, a $(R \times S)$-copula pmf.
\begin{enumerate}[(a)]
	\item Suppose that, for all $(\nu_X \times \nu_Y) \in N(\overline{\pp})$, $\sum_{x \in \nu_X} p_{x\bullet } + \sum_{y \in \nu_Y} p_{\bullet y}< 1$. Then, there exists a unique $\pp \in [\overline{\pp}]$ with the requested marginal distributions;
	\item Suppose that, for all $(\nu_X \times \nu_Y) \in N(\overline{\pp})$, $\sum_{x \in \nu_X} p_{x\bullet } + \sum_{y \in \nu_Y} p_{\bullet y} \leq 1$, with $\sum_{x \in \tilde{\nu}_X} p_{x\bullet } + \sum_{y \in \tilde{\nu}_Y} p_{\bullet y} = 1$ for some $(\tilde{\nu}_X \times \tilde{\nu}_Y) \in N(\overline{\pp})$. If, for all $(\tilde{\nu}_X \times \tilde{\nu}_Y) \in N(\overline{\pp})$ such that $\sum_{x \in \tilde{\nu}_X} p_{x\bullet } + \sum_{y \in \tilde{\nu}_Y} p_{\bullet y} = 1$, $(\Ss_X \backslash \tilde{\nu}_X \times \Ss_Y \backslash \tilde{\nu}_Y) \in N(\overline{\pp})$, then there exists a unique $\pp \in [\overline{\pp}]$ with the requested marginal distributions.
\end{enumerate}
\end{theorem}
\begin{proof} See \cite{Brossard17}, Theorems 3.1 and 3.2. \end{proof}
This result guarantees the existence of the requested distribution $\pp$ in `easy' cases, akin to cases $(a)$ and $(b(i))$ in Theorem \ref{thm:coppmf}: no zeros in $\overline{\pp}$, or zeros not lying on rows and columns carrying large target marginal weights, or `disconnected' copula pmf $\overline{\pp}$. It does not say that such a $\pp$ does not exist in the other cases. In fact, such distribution may exist, as evidenced by (\ref{eqn:decompmat}) in the Bernoulli case. However, reconstructing $\pp$ then is not achieved through the transformation (\ref{eqn:gphipsi}), as some zeros of $\overline{\pp}$ must be turned back into a positive probability, hence $\pp \not\sim \overline{\pp}$. Indeed $\pp$ belongs to an orbit $[\pp]$ of which $\overline{\pp}$ is only a limit point. 
Whether or not there exist general existence and uniqueness results in those cases remains an open question; however, a geometric perspective similar to Remark \ref{rmk:unik} suggests positive conclusions.  

\ppn Note that the IPF procedure is not tied to uniform margins and can be used for identifying the distribution $\pp$ with any requested sets of margins on the nucleus $[\overline{\pp}]$.

\subsection{Yule's coefficient} 

Inspired by Spearman's $\rho$, one can define a margin-free measure of overall concordance in $\pp$ as Pearson's correlation coefficient computed on $\overline{\pp}$. By analogy with Section \ref{subsec:Yule2}, we call such a coefficient Yule's coefficient $\Upsilon$. 
Suppose that $U$ is Discrete uniform on $\{\frac{1}{R+1},\frac{2}{R+1},\ldots,\frac{R}{R+1}\}$, $V$ is Discrete uniform on $\{\frac{1}{S+1},\frac{2}{S+1},\ldots,\frac{S}{S+1}\}$, and their joint pmf is given by the copula pmf $\overline{\pp}$. It can then be checked that Pearson's correlation between $U$ and $V$ is 
\begin{equation} \Upsilon =  3 \sqrt{\frac{(R-1)(S-1)}{(R+1)(S+1)}} \left(\frac{4}{(R-1)(S-1)} \sum_{u=0}^{R-1} \sum_{v = 0}^{S-1} uv\bar{p}_{uv}-1 \right). \label{eqn:Yule}\end{equation}

\ppn This coefficient is equal to $1$ or $-1$ if and only if $\overline{\pp}$ is a diagonal matrix, which obviously requires $R = S$. The diagonal copula pmf's of size $(R \times R)$ are clearly 
\begin{equation} \overline{\mm} \doteq  \begin{pmatrix} \frac{1}{R} & & 0 \\ & \ddots & \\ 0 & & \frac{1}{R} \end{pmatrix} \quad \text{ or } \quad   \overline{\ww} \doteq \begin{pmatrix} 0 & & \frac{1}{R} \\ & \iddots & \\ \frac{1}{R} & & 0 \end{pmatrix}, \label{eqn:Frech} \end{equation}
the Fr\'echet bounds analogous to (\ref{eqn:FrechBerncoppmf}). Note that any $\pp \in \Ps_{R \times R}$ represented by a diagonal matrix is easily seen to admit $\overline{\mm}$ or $\overline{\ww}$ as copula pmf. Those fall into case $(b\,(i))$ of Theorem \ref{thm:coppmf} and correspond to (positive or negative) `absolute association'; cf.\ Section \ref{sec:structzeros}. There also exist non-diagonal distributions $\pp \in \Ps_{R \times R}$, belonging to case $(b\,(ii))$ of Theorem \ref{thm:coppmf}, which admit $\overline{\mm}$ or $\overline{\ww}$ as copula pmf as well. Those would be akin to `complete association', by analogy to Section \ref{sec:structzeros}. As in the bivariate Bernoulli case, the distinction between `absolute' and `complete association' is only a marginal feature which must be ignored by the copula. For instance, positive `absolute association', i.e.\ probability weight concentrated on the main diagonal of $\pp$, is only possible if $\pp_X \equiv \pp_Y$. By contrast, dependence as strong as can be between unequal discrete marginals must turn into `complete association'. As a result, $\Upsilon = \pm 1$ without distinction between `absolute' and `complete' association.

\ppn Now, when $R \neq S$, $|\Upsilon|$ cannot reach 1. Indeed, if $X$ and $Y$ do not take the same number of values, it is hard to conceive a sense of `perfect dependence', as the associated copula pmf can never approach any of the diagonal forms $\overline{\mm}$ or $\overline{\ww}$. In that case, the maximum value attained by $|\Upsilon|$ occurs when $\pp$ is the pmf associated to the Fr\'echet bounds in the class of $(R \times S)$-bivariate discrete distributions with uniform margins \citep[`Exemple I']{Frechet51}.

\ppn It is also clear that, $\Upsilon$ being essentially the discrete analogue of Spearman's $\rho$, it only detects monotonic dependence (`concordance') between $X$ and $Y$. In particular, for $\max(R,S)>2$, $\Upsilon$ can be 0 even when $X$ and $Y$ are not independent. Genuine measures of dependence $\Delta$, in the sense of $\Delta = 0 \iff X\indep Y$, may be defined along the same way as in \cite{Geenens18}. 

\begin{example} The copula pmf being, by definition, margin-free, its construction only requires a sense of order for the `values' of $X$ and $Y$. In particular, if $X$ and/or $Y$ are ordinal random variables, then it remains meaningful to construct their copula pmf in order to understand their dependence. This is illustrated here through data on congenital sex organ malformations cross-classified by maternal alcohol consumption from a study described in \cite{Graubard87}:
\begin{equation*} \begin{array}{l || c c c c c c | c}
  &  \multicolumn{5}{c}{\text{Maternal Alcohol Consumption (drinks/day)}} & & \\
  &  0 & <1 & 1-2 & 3-5 & \geq 6 & & \\
\hline\hline
 \text{No Malformation} &  17,066 & 14,464 & 788 & 126 & 37 & & 32,481 \\ 
 \text{Malformation} & 48 & 38 & 5 & 1 & 1 & & 93  \\
\hline 
 & 17,114 & 14,502 & 793 & 127 & 38 & & 32,574 \\
\end{array}.  \end{equation*}


\ppn Agreed that `No malformation' $\prec$ `Malformation', the dependence between maternal alcohol consumption and congenital malformation can be understood through that of the bivariate discrete distribution 
\begin{equation} \begin{array}{c l ||c c c c c c c | c}
& \bsfrac{Y}{X} & & 0 & 1 & 2 & 3 & 4 & & \\
\hline\hline
& 0 & & 0.52391 &  0.44404 & 0.02419 & 0.00387 & 0.00114  & & 0.99714 \\ 
& 1 & & 0.00147 & 0.00117 & 0.00015 & 0.00003 & 0.00003 & & 0.00286  \\
\hline 
& & & 0.52539 & 0.44520 & 0.02434 & 0.00390 & 0.00117  & & 1 \\
\end{array},  \label{eqn:TabGooddist}\end{equation}
whose dependence structure is not easily apparent. The IPF procedure (Section \ref{sec:IPF}) returns the copula pmf of this distribution:
\[\begin{pmatrix} 0.137&  0.140 & 0.098 & 0.087 & 0.037 \\ 0.063 & 0.060 & 0.102 & 0.113 & 0.163 \end{pmatrix}, \]
which can be displayed as the confetti plot shown in Figure \ref{fig:Tab21Agrdrist}.
\begin{figure}[h]
\centering
\includegraphics[width=.5\textwidth,trim=0 2cm 0 0]{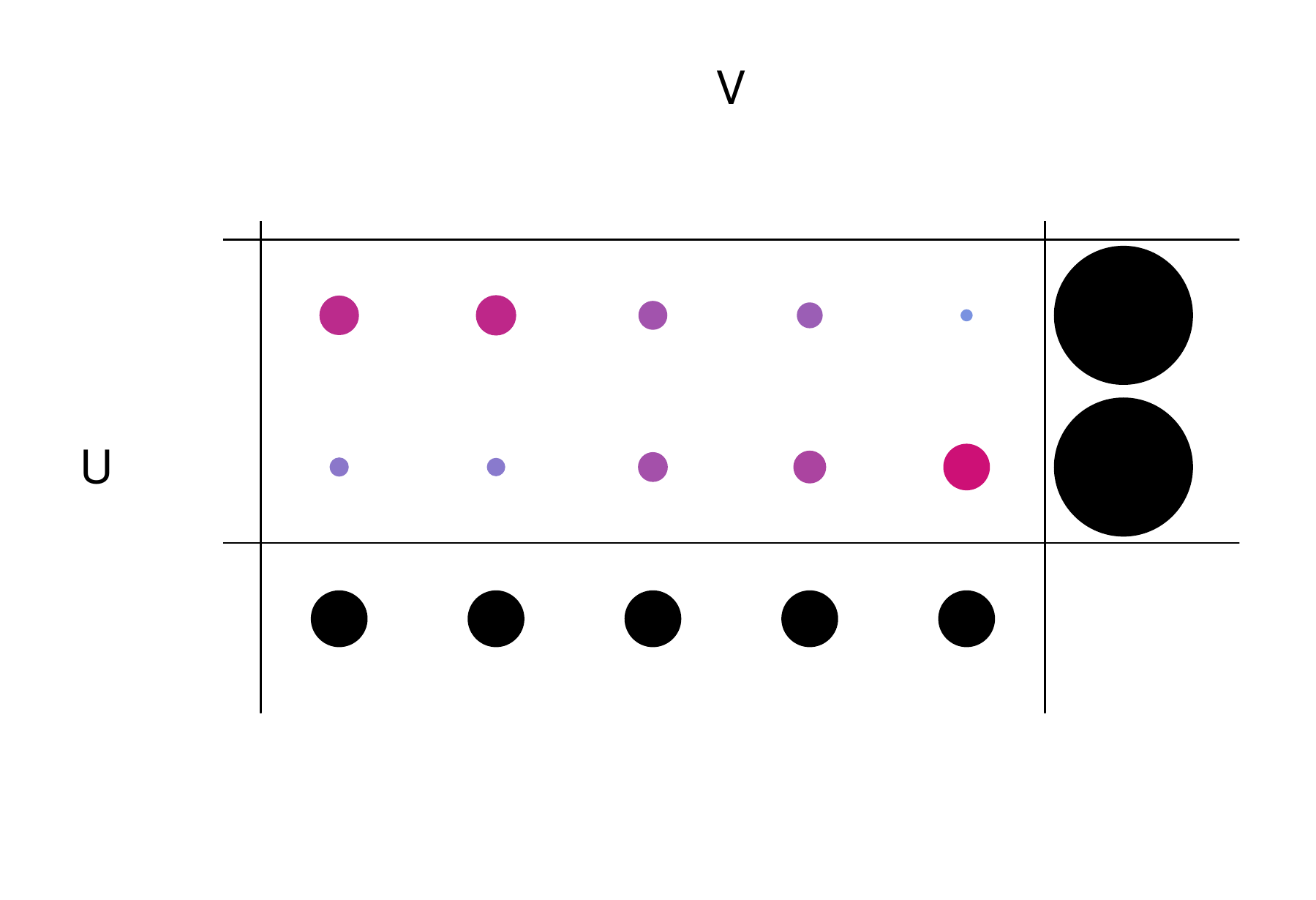}
\caption{Confetti plot of the copula pmf $\overline{\pp}$ for distribution (\ref{eqn:TabGooddist}).}
\label{fig:Tab21Agrdrist}
\end{figure}
The adverse effect of maternal alcohol consumption on the risk of congenital malformation now appears clearly, and is quantified by a positive value of Yule's coefficient of $\Upsilon = 0.358$. Entirely `margin-free', such a copula-based measure of association between two ordinal random variables does not rely on assigning scores to each category as is otherwise necessary \citep[Section 2.3]{Kateri14} --  note that the choice $X \in \{0,1\}$ and $Y \in \{0,1,2,3,4\}$ in (\ref{eqn:TabGooddist}) has no impact whatsoever. This seems desirable, as \citet[p.\,740]{Goodman54} noted: ``{\it We feel that the use of arbitrary scores to motivate measures is infrequently appropriate}."
\end{example}

\section{Parametric discrete copulas} \label{sec:paramcop}

Paralleling the continuous case, one can construct parametric models of copula pmf's. In fact, any parametric continuous copula as in Definition \ref{dfn:classcop} readily gives rise to a discrete copula pmf of any dimension $(R \times S)$, as described in Section \ref{sec:discrclasscop}. One may also think of specific discrete copulas originating from particular bivariate discrete distributions, such as the Binomial copula (Section \ref{sec:binomcop}) or truncated Geometric copula (Section \ref{sec:truncgeomcop}). Finally, some models for discrete copulas may arise naturally as well from direct specifications of the odds ratio matrix (Section \ref{sec:Goodcop}).

\subsection{Discrete versions of classical continuous copulas} \label{sec:discrclasscop}

Let $C$ be a continuous copula as in Definition \ref{dfn:classcop}. For $u \in \{0,\ldots,R-1\}$ and $v \in \{0,\ldots,S-1\}$, define
\begin{equation} \overline{p}_{uv} = C\left(\frac{u+1}{R},\frac{v+1}{S}\right)-C\left(\frac{u}{R},\frac{v+1}{S}\right)-C\left(\frac{u+1}{R},\frac{v}{S}\right)+C\left(\frac{u}{R},\frac{v}{S}\right). \label{eqn:disccontcop} \end{equation}
Then, as $C$ has uniform margins on $\Is$, it follows, for any $u,v$,
\[\sum_{v = 0}^{S-1} \overline{p}_{uv} = \frac{1}{R}\qquad \text{ and } \qquad \sum_{u = 0}^{R-1} \overline{p}_{uv} = \frac{1}{S}.\]
Hence the $(R \times S)$-discrete distribution
\[\overline{\pp} = [\overline{p}_{uv} ]_{\substack{u = 0,\ldots,R-1, \\ v = 0,\ldots,S-1}} \]
is a copula pmf as defined by Definition \ref{dfn:discrcop}. It is thus straightforward to define a $(R \times S)$-discrete version of the classical continuous copulas such as Gaussian, Student, Frank, Clayton or Gumbel, to cite a few.

\ppn For simple parametric continuous copulas $C$, such $\overline{\pp}$ can be written in closed form. For instance, it can be checked that the $(R \times S)$-discrete version of the FGM copula is, for $\theta \in [-1,1]$,
\begin{equation} \overline{p}_{uv} = \frac{1}{RS}\left(1+ \theta (1-\frac{2u+1}{R})(1-\frac{2v+1}{S}) \right), \qquad (u,v) \in \{0,\ldots,R-1\} \times\{0,\ldots,S-1\}. \label{eqn:discrFGM} \end{equation}
Figure \ref{fig:FGMcop} in Appendix shows confetti plots of this copula pmf for $\theta = 1$ and $(R,S) \in \{ (3,3), (5,5),(5,3)\}$.

\ppn For $\theta = 0$ in (\ref{eqn:discrFGM}), one finds $\overline{p}_{uv} = 1/RS$ for all $(u,v)$, which is the $(R \times S)$-independence copula pmf: 
\begin{equation}
 \overline{\pp} \doteq \overline{\ppi} = \frac{1}{RS} \uno_R \uno_S^T,   \label{eqn:RSindepcop}
\end{equation}
with $\uno_k=(\underbrace{1,\ldots,1}_{k \text{ times}})^T$.

\ppn For illustration, Figures \ref{fig:Claytposcop}, \ref{fig:Claytnegcop}, \ref{fig:Stucop}, \ref{fig:Gumbcop} in Appendix show the discrete copula pmf's obtained from the continuous Clayton copula with $\theta = 0.8$ and $\theta = -0.8$, from the continuous Student copula with $d=1$ and $\rho=0$ and from the continuous Gumbel copula with $\theta =2$, for $(R,S) \in \{(3,3),(5,5),(5,3)\}$. 

\begin{remark} \label{rmk:discrcontcop} The discrete copula pmf's derived from a continuous one through (\ref{eqn:disccontcop}) are obtained by overlaying $C$ on the regular mesh $\{0,\frac{1}{R},\ldots,\frac{R-1}{R},1\} \times \{0,\frac{1}{S},\ldots,\frac{S-1}{S},1\}$ over the unit square $\Is$. The switch from continuous to discrete is thus carried out `in the copula world', keeping all marginals uniform. The so-produced discrete copula pmf's $\overline{\pp}$ can then be used in a second time for modelling dependence between discrete random variables and/or constructing new bivariate discrete distributions with given marginals, as per Section \ref{sec:arbdistrRS}. The idea of combining two distinct building blocks, the marginals on one side and the dependence structure on the other, is maintained. By contrast, when apprehending a bivariate discrete distribution $F_{XY}$ through (\ref{eqn:Sklar}) with a certain continuous copula $C$, the switch from continuous to discrete occurs when overlaying $C$ directly on the mesh $\Ran F_X \times \Ran F_Y$ set by the marginal distributions of $F_{XY}$. Taking two steps in one, this explains why, in such models, dependence and marginal distributions always get confused, in contradiction with the essence of copula modelling. It is noteworthy, though, that the two approaches coincide in the case of a continuous vector $(X,Y)$, as then the mesh $\Ran F_X \times \Ran F_Y$ reduces down to the whole unit square $\Is$, the `continuous regular mesh' in some sense. \qed
\end{remark}

\subsection{The Binomial copula} \label{sec:binomcop}

Let $(X_1,Y_1)$, $\ldots$, $(X_n,Y_n)$ be independent copies of a bivariate Bernoulli random variable with pmf (\ref{eqn:bivbernpmf}). \citet[Section 3]{Marshall85} defined the bivariate Binomial as the distribution of the vector $(X,Y)=\left(\sum_{i=1}^n X_i,\sum_{i=1}^n Y_i\right)$, by strict analogy with a (univariate) Binomial distribution being the sum of $n$ independent replications of a Bernoulli random variable. This bivariate Binomial is parameterised by $n$ and the matrix (\ref{eqn:matBivBern}), and its pmf is, for $(x,y) \in \{0,\ldots,n\} \times \{0,\ldots,n\}$, 
\[\P(X = x,Y = y) = \sum_{k = \max(x+y-n,0)}^{\min(x,y)} \binom{n}{k,x-k,y-k,n-x-y+k} p_{00}^{n-x-y+k}p_{10}^{x-k} p_{01}^{y-k} p_{11}^k. \]
Then, it can be checked that the odds-ratios (\ref{eqn:omxy}) are
\[\omega_{xy} = \sum_{k = \max(x+y-n,0)}^{\min(x,y)} \frac{ \binom{n}{k,x-k,y-k,n-x-y+k}}{ \binom{n}{x}\binom{n}{y}}\  \omega^k, \]
where $\omega = \frac{p_{00}p_{11}}{p_{01}p_{10}}$ is the odds-ratio of the initial bivariate Bernoulli. For $n$ fixed, the dependence structure in a bivariate Binomial is thus only driven by one parameter $\omega$, and the corresponding Binomial$(n)$-copula, which is a $((n+1)\times (n+1))$-discrete distribution with uniform margins, is a one-parameter model. 

\ppn For instance, if $n=2$, the bivariate Binomial distribution is 
identified to the matrix
\begin{equation*} \pp = \begin{pmatrix}  p_{00}^2 & 2p_{00}p_{01} & p_{01}^2 \\ 2p_{00}p_{10} &  2(p_{11}p_{00}+p_{10}p_{01}) & 2p_{11}p_{01} \\ p_{10}^2 & 2p_{10}p_{11} & p_{11}^2  \end{pmatrix} \in \Ps_{3 \times 3}, \label{eqn:bivbinomat} \end{equation*}
and the corresponding odds ratio matrix (\ref{eqn:ORmat}) is  
\begin{equation*} \Omega(\pp) = \begin{pmatrix}  \frac{1}{2}(\omega+1) & \omega \\  \omega & \omega^2 \end{pmatrix}. \label{eqn:bivbinoORM} \end{equation*}
Now define the `completed' odds-ratio matrix
\begin{equation} \widetilde{\Omega}(\pp) = \begin{pmatrix} 1 & \uno^T_{2} \\ 
\uno_{2} & \Omega(\pp) \end{pmatrix}, \label{eqn:compORmat}\end{equation}
which includes the trivial odds-ratios $\omega_{00}, \omega_{0y}$'s and $\omega_{x0}$'s all equal to 1. If $\widetilde{\Omega}(\pp)$ was normalised so as to have unit $L_1$-norm, it would evidently be a bivariate pmf with odds ratio matrix $\Omega(\pp)$, hence sharing the same copula with $\pp$. It is easier to make the marginals of such a matrix $\widetilde{\Omega}$ uniforms, rather than those of $\pp$, given its simple form. 
Here, one obtains (after some algebra) by making the margins of (\ref{eqn:compORmat}) into uniforms through (\ref{eqn:gphipsi}):
\begin{equation*} \overline{\pp} = \frac{1}{3} \begin{pmatrix} \frac{\omega(\omega+1)}{\omega^2 + \omega + 1 + \sqrt{\omega(\omega+2)(2\omega+1)}} & \frac{\sqrt{\omega(\omega+2)(2\omega+1)}-3\omega}{(\omega-1)^2} & \frac{\omega+1}{\omega^2 + \omega + 1 + \sqrt{\omega(\omega+2)(2\omega+1)}} \\
\frac{\sqrt{\omega(\omega+2)(2\omega+1)}-3\omega}{(\omega-1)^2} & \frac{\omega^2+4\omega +1- 2\sqrt{\omega(\omega+2)(2\omega+1)}}{(\omega-1)^2}&  \frac{\sqrt{\omega(\omega+2)(2\omega+1)}-3\omega}{(\omega-1)^2}\\
\frac{\omega+1}{\omega^2 + \omega + 1 + \sqrt{\omega(\omega+2)(2\omega+1)}} & \frac{\sqrt{\omega(\omega+2)(2\omega+1)}-3\omega}{(\omega-1)^2} & \frac{\omega(\omega+1)}{\omega^2 + \omega + 1 + \sqrt{\omega(\omega+2)(2\omega+1)}}                 \end{pmatrix} \label{eqn:bivbinocop} \end{equation*}
for $\omega \neq 1$. For $\omega = 1$, of course, $\overline{\pp} = \overline{\ppi}$, the $(3 \times 3)$-independence copula pmf (\ref{eqn:RSindepcop}). See also that, for $\omega = 0$ or $\omega = \infty$, $\overline{\pp} = \overline{\ww}$ and $\overline{\pp} = \overline{\mm}$, the Fr\'echet lower and upper bounds (\ref{eqn:Frech}) in 3 dimensions. Indeed, from (\ref{eqn:Yule}), one has 
\[\Upsilon = \frac{\omega^2-1}{\omega^2+\omega+1+\sqrt{\omega(\omega+2)(2\omega+1)}} \]
as Yule's coefficient for this copula pmf, which is $\Upsilon =-1$ for $\omega =0$ and $\Upsilon = 1$ for $\omega = \infty$. This family of Binomial copulas is thus complete as it allows all values for Yule's coefficients from $-1$ and $1$. Confetti plots of this Binomial$(2)$ copula are given in Figure \ref{fig:binomcop} for several values of $\omega$. 

\begin{figure}[h]
\centering
\includegraphics[width=1\textwidth]{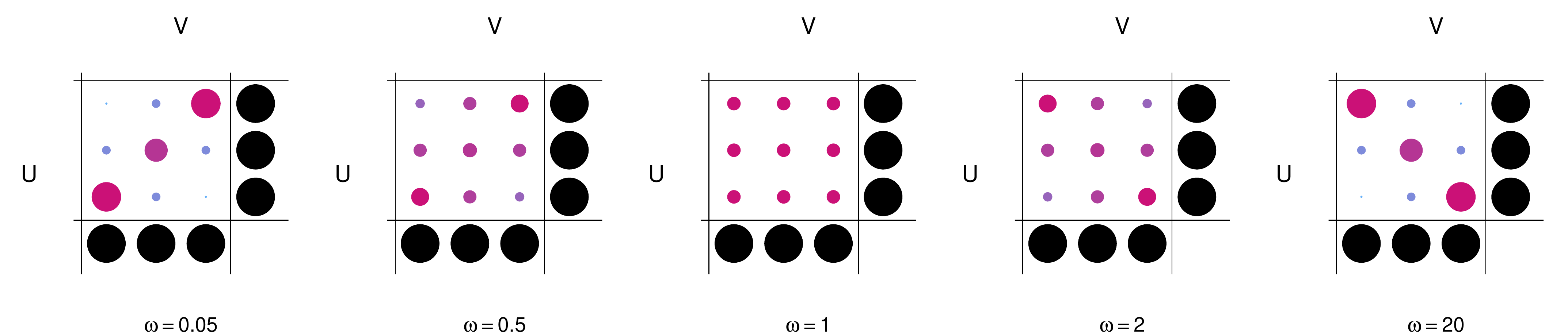}
\caption{Confetti plots of the Binomial$(2)$ copula pmf for $\omega = 0.05, 0.5, 1, 2$ and $20$.}
\label{fig:binomcop}
\end{figure}


\subsection{The truncated geometric copula} \label{sec:truncgeomcop}

Let $(X_1,Y_1)$, $(X_2,Y_2)$, $\ldots$ again be a sequence of independent replications from the bivariate Bernoulli distribution (\ref{eqn:bivbernpmf}). \citet[Section 6]{Marshall85} defined the bivariate Geometric distribution as the distribution of the vector $(X,Y)$ where $X$ is the number of $0$'s before the first 1 in the sequence $X_1,X_2,\ldots$, and $Y$ the number of $0$'s before the first 1 in the sequence $Y_1,Y_2,\ldots$. The pmf is
\begin{equation} \P(X=x,Y = y) = \left\{\begin{array}{l l} p_{00}^x p_{10} p_{\bullet 0}^{y-x-1} p_{\bullet 1} & 0 \leq x < y \\
p_{00}^x p_{11} & 0 \leq x=y \\
p_{00}^y p_{01} p_{ 0 \bullet}^{x-y-1} p_{1 \bullet } & 0 \leq y < x
\end{array} \right. . \label{eqn:geompmf} \end{equation}

\ppn Now, consider the truncated random variables $\tilde{X}_N = \min(X,N-1)$ and $\tilde{Y}_N = \min(Y,N-1)$, for some $N \geq 2$. The random vector $(\tilde{X}_N,\tilde{Y}_N)$ has pmf
\begin{equation} \P(\tilde{X}_N=x,\tilde{Y}_N = y) =  \left\{\begin{array}{l l}  \P(X=x,Y=y) & 0 \leq x,y <N-1 \\ 
p_{00}^y p_{01} p_{0 \bullet}^{N-y-2} & x = N-1, 0  \leq y < N-1 \\
p_{00}^x p_{10} p_{\bullet 0 }^{N-x-2} & 0 \leq x < N-1,  y = N-1 \\
p_{00}^{N-1} & x = N-1, y=N-1
\end{array} \right. . \label{eqn:truncgeodistr} \end{equation}

For instance, for $N = 3$, one obtains 
\begin{equation*} \begin{array}{c || c c c |c}
 \bsfrac{\tilde{Y}_N}{\tilde{X}_N} & 0 & 1 & 2 & \\
\hline 
 0 & p_{11} & p_{10}p_{\bullet 1} & p_{10} p_{\bullet 0} & p_{1\bullet} \\
 1 & p_{01} p_{1 \bullet} &  p_{00}p_{11} & p_{00}p_{10} & p_{0\bullet} p_{1\bullet}\\
 2 & p_{01} p_{0\bullet} & p_{00}p_{01} & p_{00}^2 & p_{0\bullet}^2 \\
\hline
  & p_{\bullet 1} &  p_{\bullet 0} p_{\bullet 1} & p_{\bullet 0}^2 & 
\end{array}, \label{eqn:bivtruncGeoopmf} \end{equation*}
identified to the matrix 
\[\pp = \begin{pmatrix}  p_{11} & p_{10}p_{\bullet 1} & p_{10} p_{\bullet 0} \\ p_{01} p_{1 \bullet} &  p_{00}p_{11} & p_{00}p_{10} \\ p_{01} p_{0\bullet} & p_{00}p_{01} & p_{00}^2  \end{pmatrix} \in \Ps_{3 \times 3}.\]

\ppn One can check that the odds ratio matrix (\ref{eqn:ORmat}) is here 
\begin{equation} \Omega(\pp) = \begin{pmatrix}  \omega \frac{p_{11}}{p_{1\bullet} p_{\bullet 1}} & \omega \frac{p_{10}}{p_{1\bullet} p_{\bullet 0}} \\  \omega \frac{p_{01}}{p_{0\bullet} p_{\bullet 1}} & \omega \frac{p_{00}}{p_{0\bullet} p_{\bullet 0}} \end{pmatrix}, \label{eqn:bivtruncGeooORM} \end{equation}
where $\omega = \frac{p_{00}p_{11}}{p_{01}p_{10}}$ is again the odds-ratio of the initial bivariate Bernoulli. As opposed to the binomial case, the structure of dependence in this bivariate Geometric depends not only on $\omega$, but also on the marginals of the initial bivariate Bernoulli through $p_{\bullet 1}$ and $p_{1\bullet}$ (the other quantities in (\ref{eqn:bivtruncGeooORM}) are functions of $\omega$, $p_{\bullet 1}$ and $p_{1 \bullet}$; Section \ref{sec:arbbern}). It is thus a 3-parameters copula family.

\ppn One can define a `standard' version of it by fixing $p_{\bullet 1} = p_{1 \bullet} = 1/2$, that is, assuming that the initial bivariate Bernoulli distribution is the copula pmf (\ref{eqn:Berncoppmf}). Then one recovers a one parameter copula model, driven by the odds ratio matrix
\begin{equation} \Omega(\pp) = \begin{pmatrix}  \omega \frac{2\sqrt{\omega}}{1+\sqrt{\omega}} & \omega \frac{2}{1+\sqrt{\omega}} \\  \omega \frac{2}{1+\sqrt{\omega}} & \omega \frac{2\sqrt{\omega}}{1+\sqrt{\omega}} \end{pmatrix}. \label{eqn:omtruncgeo}
\end{equation}
Making the marginals of the corresponding `completed' odds ratio matrix (\ref{eqn:compORmat}) into uniforms through (\ref{eqn:gphipsi}), one obtains (after some algebra) the explicit form of the copula pmf:
\begin{equation} \overline{\pp} = \frac{1}{3} \begin{pmatrix} \frac{2\omega}{2\omega+\sqrt{8\omega +1}+1} & \frac{\sqrt{8\omega+1}+1}{2(2\omega+\sqrt{8\omega +1}+1)} & \frac{\sqrt{8\omega+1}+1}{2(2\omega+\sqrt{8\omega +1}+1)} \\
\frac{\sqrt{8\omega+1}+1}{2(2\omega+\sqrt{8\omega +1}+1)} & \frac{\sqrt{\omega}(\sqrt{8\omega+1}+1)^2}{4(\sqrt{\omega}+1)(2\omega+\sqrt{8\omega +1}+1)} &  \frac{4\omega-1-\sqrt{8\omega+1}}{4(\sqrt{\omega}-1)(\sqrt{\omega}+1)^2}\\
\frac{\sqrt{8\omega+1}+1}{2(2\omega+\sqrt{8\omega +1}+1)} & \frac{4\omega-1-\sqrt{8\omega+1}}{4(\sqrt{\omega}-1)(\sqrt{\omega}+1)^2} & \frac{\sqrt{\omega}(\sqrt{8\omega+1}+1)^2}{4(\sqrt{\omega}+1)(2\omega+\sqrt{8\omega +1}+1)} \end{pmatrix} \label{eqn:bivtruncGeocop} \end{equation}
for $\omega \neq 1$. If $\omega =1$, then $\overline{\pp} = \overline{\ppi}$. If $\omega =\infty$, then $\overline{\pp} = \overline{\mm}$, and $\Upsilon = 1$. On the other hand, when $\omega = 0$, then 
\[ \overline{\pp} = \begin{pmatrix} 0 & 1/6 & 1/6 \\ 1/6 & 0 & 1/6 \\ 1/6 & 1/6 & 0 \end{pmatrix},\]
and from (\ref{eqn:Yule}) it can easily be seen that $\Upsilon = -1/2$ in that case. We call (\ref{eqn:bivtruncGeocop}) the truncated Geometric$(3)$ copula pmf, confetti plots of which is given in Figure \ref{fig:truncgeocop} for several values of $\omega$.

\begin{figure}[h]
\centering
\includegraphics[width=1\textwidth]{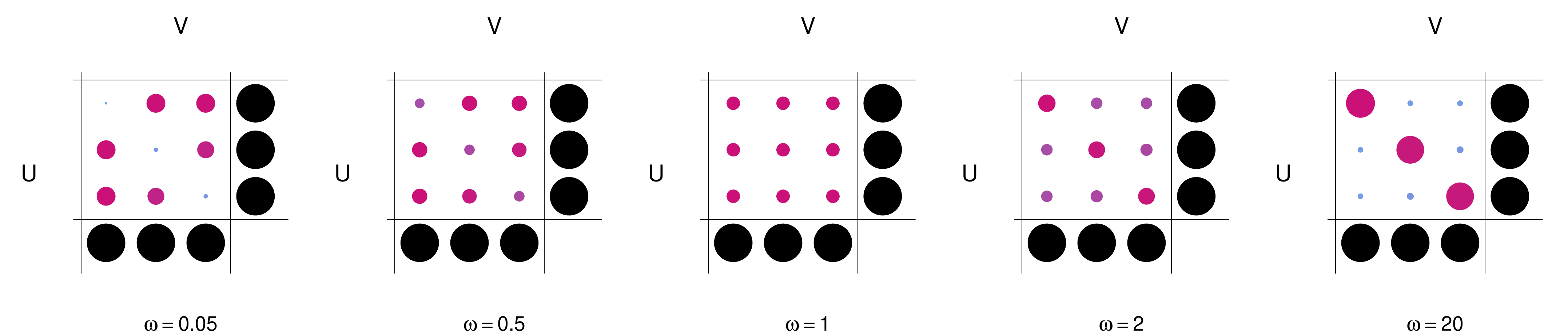}
\caption{Confetti plots of the truncated Geometric$(3)$ copula pmf for $\omega = 0.05, 0.5, 1, 2$ and $20$.}
\label{fig:truncgeocop}
\end{figure}

\ppn Naturally, any threshold value $N \geq 2$ can be considered, yielding the truncated Geometric$(N)$ copula pmf for $(N\times N)$-distributions.

 \subsection{The Goodman copula} \label{sec:Goodcop}
 
\cite{Goodman79} suggested a model of association in a contingency table with ordered categories, which can naturally be applied to a bivariate discrete vector as well. It relies on assuming that, in $\pp \in \Ps_{R \times S}$, all `local' odds-ratio $p_{x,y}p_{x-1,y-1}/(p_{x,y-1}p_{x-1,y})$, $(x,y) \in \Ss_X\backslash\{0\} \times \Ss_Y\backslash\{0\}$, are constant and equal to some $\theta > 0$. This leads to odds-ratios (\ref{eqn:omxy}) of the form
\[\omega_{xy} = \theta^{xy},\quad \forall (x,y) \in \Ss_X\backslash\{0\} \times \Ss_Y\backslash\{0\}.  \]
In the case $R=S=3$, this yields the odds-ratio matrix
\begin{equation*} \Omega(\pp) =\begin{pmatrix}
\theta & \theta^2 \\
\theta^2 & \theta^4
\end{pmatrix}.
\label{eqn:GoodmanORM} \end{equation*} 
By making the marginals of the `completed' odds-ratio matrix (\ref{eqn:compORmat}) into uniforms through (\ref{eqn:gphipsi}), one obtains (after some algebra) the copula pmf
\begin{equation} \overline{\pp} = \frac{1}{3} \begin{pmatrix}
\frac{2\theta^2}{\theta(2\theta-1) + 2+\sqrt{\theta(4\theta^2+\theta+4)}} & \frac{2\sqrt{\theta}}{3\sqrt{\theta}+\sqrt{4\theta^2+\theta+4}} & \frac{2}{\theta(2\theta-1) + 2+\sqrt{\theta(4\theta^2+\theta+4)}} \\
\frac{2\sqrt{\theta}}{3\sqrt{\theta}+\sqrt{4\theta^2+\theta+4}} & \frac{\theta^2+\theta +1- \sqrt{\theta(4\theta^2+\theta+4)}}{(\theta-1)^2}&  \frac{2\sqrt{\theta}}{3\sqrt{\theta}+\sqrt{4\theta^2+\theta+4}} \\
\frac{2}{\theta(2\theta-1) + 2+\sqrt{\theta(4\theta^2+\theta+4)}} & \frac{2\sqrt{\theta}}{3\sqrt{\theta}+\sqrt{4\theta^2+\theta+4}} & \frac{2\theta^2}{\theta(2\theta-1) + 2+\sqrt{\theta(4\theta^2+\theta+4)}}
\end{pmatrix}
\label{eqn:Goodmancop} \end{equation}
for $\theta \neq 1$. For $\theta = 1$, we have $\overline{\pp} = \overline{\ppi}$. See also that, for $\theta = 0$, $\overline{\pp} = \overline{\ww}$, and for $\theta = \infty$, $\overline{\pp} = \overline{\mm}$, the Fr\'echet lower and upper bounds (\ref{eqn:Frech}). We call (\ref{eqn:Goodmancop}) the Goodman$(3,3)$ copula, of which confetti plots are given in Figure \ref{fig:Goodcop} for several values of $\theta$. It is straightforward to generalise this model to any Goodman$(R,S)$ copula, $R,S \geq 3$.

\begin{figure}[h]
	\centering
	\includegraphics[width=1\textwidth]{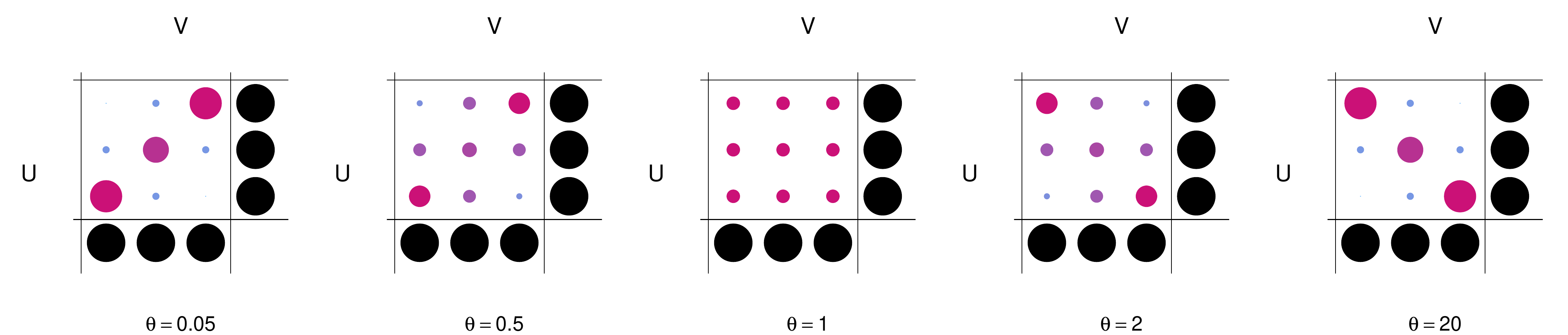}
	\caption{Confetti plots of the Goodman$(3,3)$ copula pmf for $\theta = 0.05, 0.5, 1, 2$ and $20$.}
	\label{fig:Goodcop}
\end{figure}

\section{Discrete distributions with infinite support} \label{sec:copinfsup}

\subsection{General case}

The above construction carries over to the case of discrete random variables with infinite support, say $\N$ (without loss of generality). A way of approaching this is to first consider truncated versions of the random variables of interest, like the truncated Geometric variables in Section \ref{sec:truncgeomcop}, and then let the truncation threshold $N$ tend to $\infty$. Specifically, consider a bivariate discrete vector $(X,Y)$ supported on $\N \times \N$, and assume here that $p_{xy} = \P(X=x,Y=y) >0$, $\forall (x,y) \in \N \times \N$. Let $\tilde{X}_N = \min(X,N-1)$ and $\tilde{Y}_N = \min(Y,N-1)$, for some $N \geq 2$. The pmf of the `truncated' vector $(\tilde{X}_N,\tilde{Y}_N)$ is 
\[\tilde{p}_{N;xy} = \left\{\begin{array}{ll} p_{xy} & \text{ if } 0 \leq x,y < N-1 \\
\sum_{y^* \geq N} p_{xy^*} & \text{ if } 0 \leq x < N, y = N-1 \\
\sum_{x^* \geq N} p_{x^*y} & \text{ if } 0 \leq y < N, x = N-1 \\
\sum_{x^* \geq N} \sum_{y^* \geq N} p_{x^*y^*} & \text{ if } x = N-1,y=N-1 \\
0 & \text{ if } x \geq N,y \geq N
\end{array}\right. .\]
Denote $\tilde{\pp}_N$ the corresponding matrix in $\Ps_{N \times N}$. For all integer $N$, it follows from Corollary \ref{cor:coppmf} that this bivariate discrete distribution admits a unique copula pmf $\overline{\tilde{\pp}}_N \in [\tilde{\pp}_N] \cap \Cs_{N \times N}$, as all $\tilde{p}_{N;xy}$ are positive on $\{0,\ldots,N-1\} \times \{0,\ldots,N-1\}$. By Definition \ref{dfn:discrcop}, the copula pmf $\overline{\tilde{\pp}}_N$ is the pmf of a vector $(\tilde{U}_N,\tilde{V}_N)$ whose both margins are Discrete uniform on $\{\frac{1}{N+1},\frac{2}{N+1},\ldots,\frac{N}{N+1}\}$. Now, let $N \to \infty$. It is well known that $\tilde{U}_N \toL U$, where $U \sim \Us_{[0,1]}$, and similarly $\tilde{V}_N \toL V$, where $V \sim \Us_{[0,1]}$. Hence $(\tilde{U}_N,\tilde{V}_N)$ converges in law to a bivariate distribution with continuous uniform marginals, that is, a (continuous) copula as per Definition \ref{dfn:classcop}; see Theorems 1 and 2 in \cite{Kolesarova06}. So, the dependence structure of a discrete bivariate vector supported on $\N \times \N$ can be represented by a {\it unique} continuous copula.

\ppn Importantly, this unique copula  is not any of the $C$'s satisfying (\ref{eqn:Sklar}). Analogously to Remark \ref{rmk:discrcontcop}, Sklar's theorem establishes that one can reconstruct the bivariate discrete distribution $F_{XY}$ by overlaying a copula $C$ on the mesh $\Ran F_X \times \Ran F_Y$ over the unit square. Such copula is not unique and is indissociable to the marginal distributions. By contrast, the above construction singles out one unique copula which represents the `core' of $F_{XY}$ in the spirit of the marginal transformations described in Section \ref{sec:margtransnuc}. It is independent of the margins, as it is a representation of all the odds ratios $\omega_{xy}$ (\ref{eqn:omxy}) for $(x,y) \in \N_+ \times \N_+$. The bivariate discrete distribution $F_{XY}$ can thus be broken down into its marginal distributions on one hand, and its unique copula on the other, like in the continuous case. A difference is that here, the combination of the copula and the marginals is carried out by (a limiting version of) IPF (Section \ref{sec:IPF}), not by (\ref{eqn:Sklar}).
%
 
\ppn 
Interestingly, this also allows the definition of `new' continuous copulas characterising the dependence structure inside specific $(\N \times \N)$-discrete distributions, e.g., the Geometric copula (Section \ref{sec:geocop}) from the bivariate Geometric and the Poisson copula (Section \ref{sec:poisscop}) from the bivariate Poisson. 

\subsection{The Geometric copula} \label{sec:geocop}

Consider the truncated Geometric distribution $\tilde{\pp}_N$ given by (\ref{eqn:truncgeodistr}), and set $p_{1\bullet} = p_{\bullet 1} = 1/2$. For any $N \geq 2$, one gets an odds ratio matrix $\Omega(\tilde{\pp}_N)$ involving only one parameter $\omega$, like (\ref{eqn:omtruncgeo}) for $N=2$. One can then make the margins of the `completed' odds ratio matrix into uniforms, and obtain copula pmf's $\overline{\tilde{\pp}_N} \in \Cs_{N \times N}$. Figure \ref{fig:geomcoptruncpos} shows those copula pmf's for $\omega = 2$ and $N=4, 8,16$ and $32$.

\begin{figure}[h]
\centering
\includegraphics[width=\textwidth]{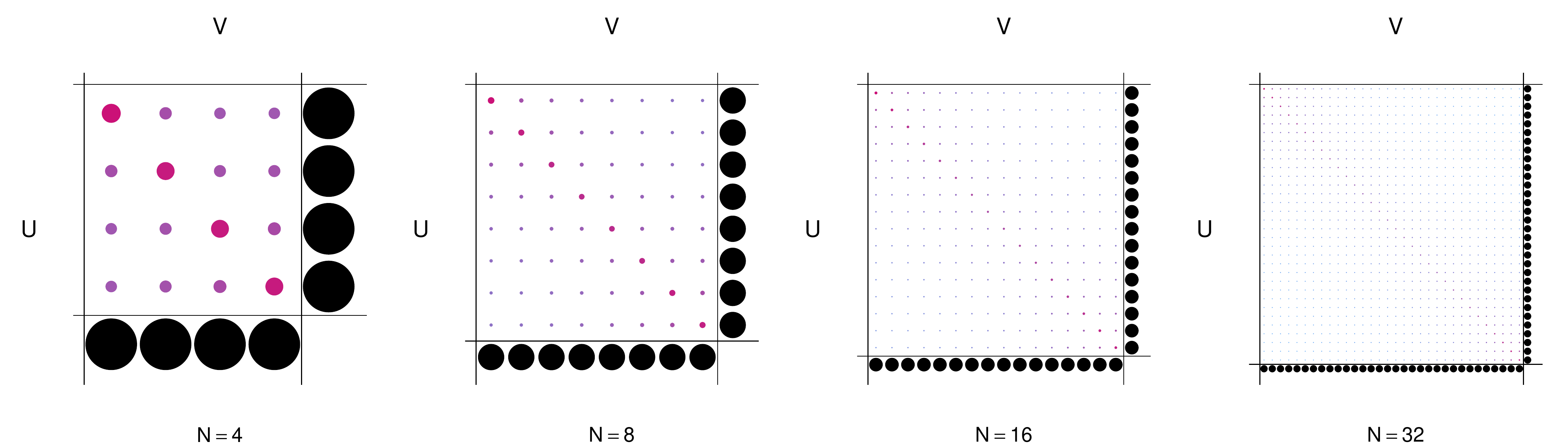}
\caption{Confetti plots of the truncated Geometric copula pmf with $\omega =2$ and growing $N$.}
\label{fig:geomcoptruncpos}
\end{figure}

\ppn In the limit $N \to \infty$, the very dense pmf turns into a continuous distribution with uniform margins as pictured in Figure \ref{fig:geomcoptruncpos2}. This copula admits a singularity along the main diagonal of the unit square $\Is$, coming from the initial geometric distribution (\ref{eqn:geompmf}) showing a different behaviour when $x=y$. The singularity reminds us of the Marshall-Olkin copula \citep[Section 3.1.1]{Nelsen06}, a link to which could be expected here given that the Marshall-Olkin bivariate Exponential distribution is the limit version of the bivariate Geometric distribution introduced above \citep[Section 6]{Marshall85}. The Geometric copula, however, remains a representative of the inner dependence structure in the purely discrete vector $(X,Y)$ whose pmf is (\ref{eqn:geompmf}), and is not the Marshall-Olkin copula. 

\begin{figure}[h] 
\begin{subfigure}[t]{0.5\textwidth}
\centering
\includegraphics[width=0.8\textwidth]{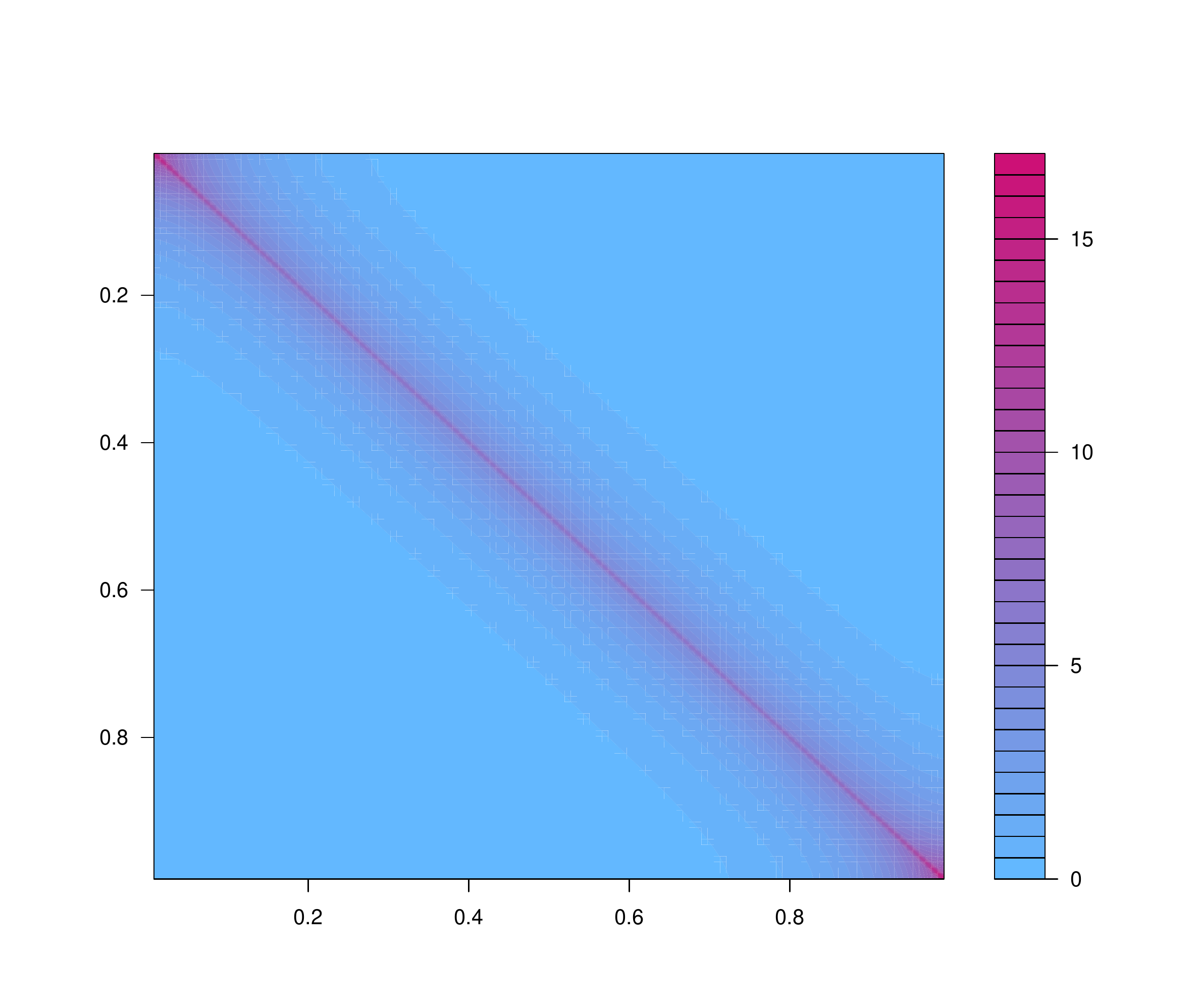}
\caption{Heat map}
\label{fig:geomtrunccopposcont}
\end{subfigure}
\hfill
\begin{subfigure}[t]{0.5\textwidth}
\centering
\includegraphics[width=\textwidth]{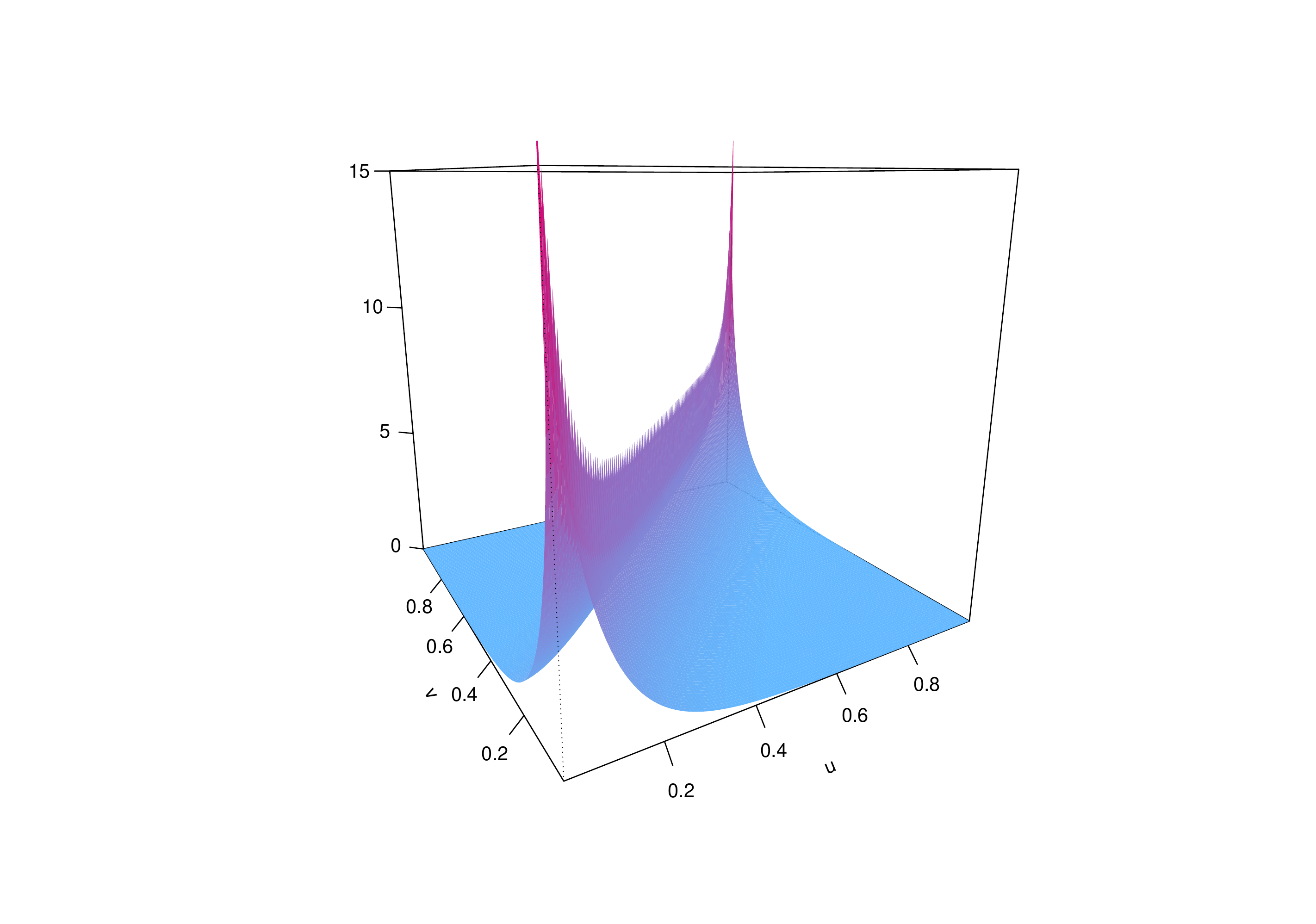}
\caption{3D-copula density}
\label{fig:geomtrunccoppospersp}
\end{subfigure}
\caption{The Geometric copula with $\omega=2$.}
\label{fig:geomcoptruncpos2}
\end{figure}

\ppn Repeating the above process of letting $N \to \infty$ with $\omega= 1/2$, the limiting Geometric copula density is seen to be identically null on the main diagonal of $\Is$, forming some sort of `inverse singularity' there. This is definitely not a Marshall-Olkin copula.

\begin{figure}[h] 
\begin{subfigure}[t]{0.4\textwidth}
\includegraphics[width=\textwidth]{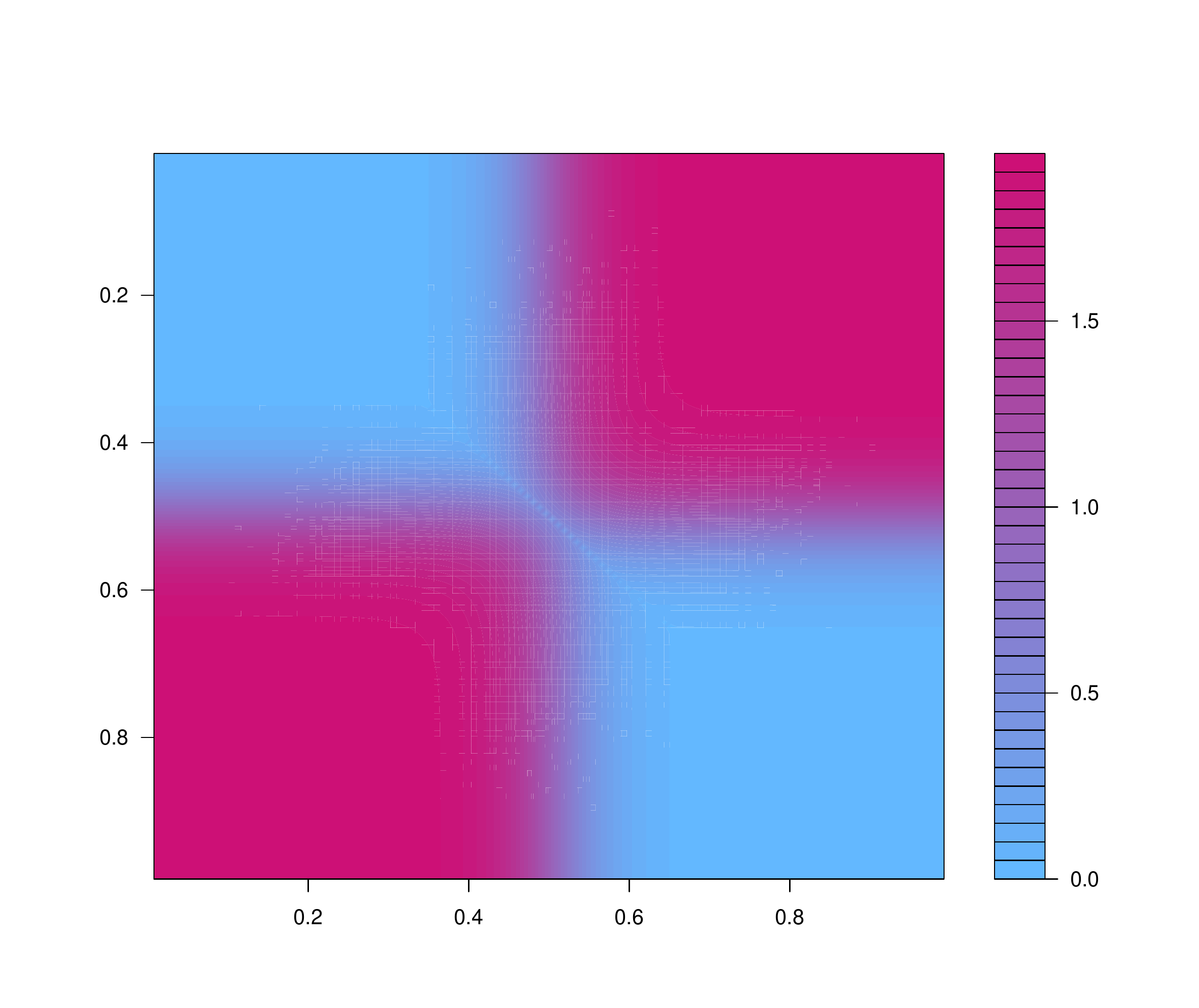}
\caption{Heat map}
\label{fig:geomtrunccopnegcont}
\end{subfigure}
\hfill
\begin{subfigure}[t]{0.5\textwidth}
\includegraphics[width=\textwidth]{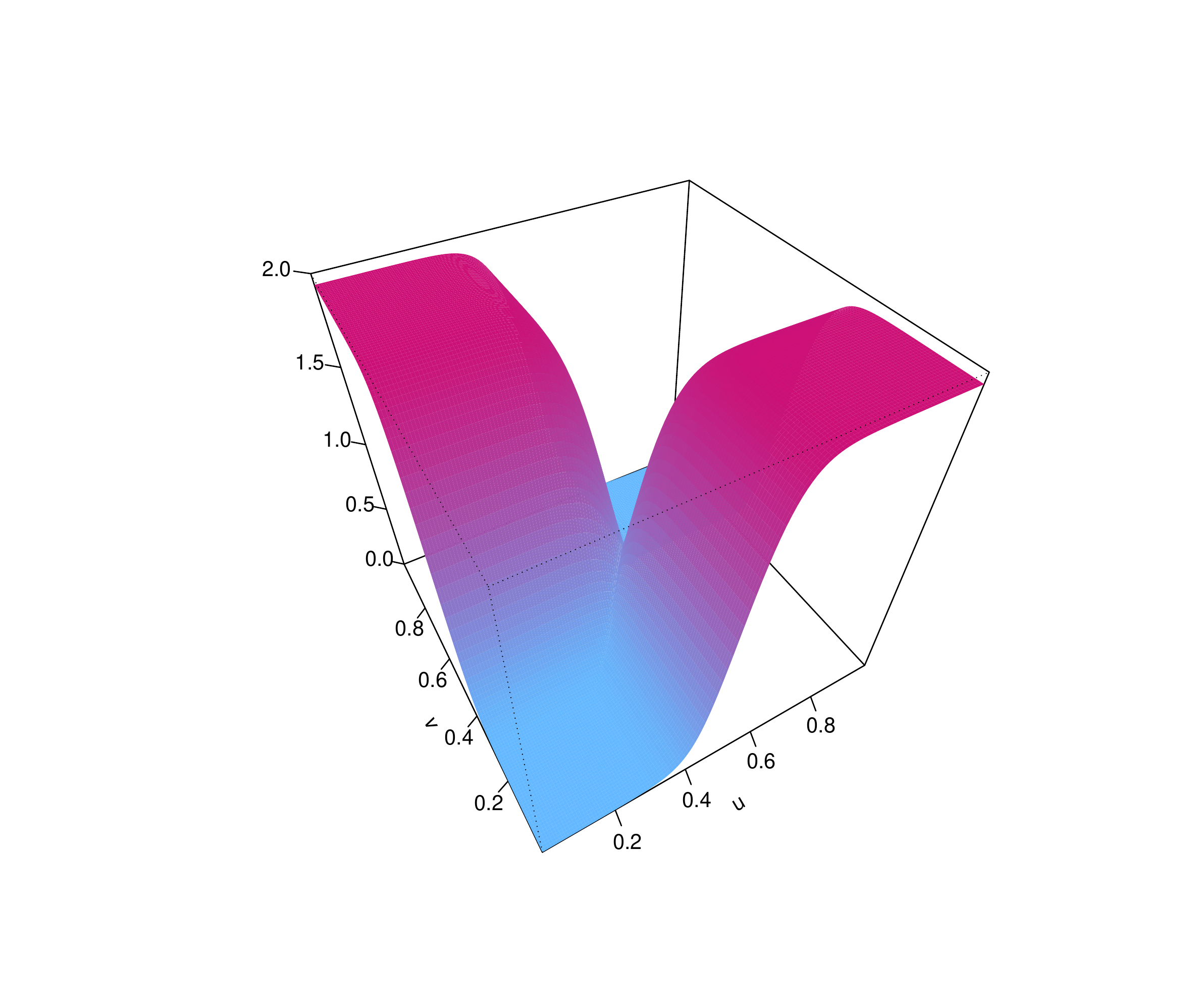}
\caption{3D-copula density}
\label{fig:geomtrunccopnegpersp}
\end{subfigure}
\caption{The Geometric copula with $\omega=1/2$.}
\label{fig:geomcoptrunc2}
\end{figure}

\subsection{The Poisson copula} \label{sec:poisscop}

Let $Z_{10} \sim \Ps(\lambda_{10})$, $Z_{01} \sim \Ps(\lambda_{01})$ and $Z_{11} \sim \Ps(\lambda_{11})$ be three independent Poisson random variables, with $\lambda_{10}, \lambda_{01} > 0$ and $\lambda_{11} \geq 0$. Then define 
\begin{equation} X = Z_{10} + Z_{11} \qquad \text{ and } \qquad Y = Z_{01} + Z_{11}. \label{eqn:trivred}\end{equation}
The distribution of the vector $(X,Y)$ is classically known as the bivariate Poisson distribution \citep{Teicher54}, parameterised by $(\lambda_{10},\lambda_{01},\lambda_{11})$. Its pmf is \citep[Section 4]{Marshall85}:
\[p_{xy}=\P(X=x,Y=y) = e^{-(\lambda_{10}+\lambda_{01}+\lambda_{11})} \frac{\lambda_{10}^x}{x!}\frac{\lambda_{01}^y}{y!} \sum_{i=0}^{\min(x,y)} i! \binom{x}{i} \binom{y}{i}  \left(\frac{\lambda_{11}}{\lambda_{10} \lambda_{01}}\right)^i,\]
for $(x,y) \in \{0,1\ldots\} \times \{0,1,\ldots\}$, and clearly $X \sim \Ps(\lambda_{10} + \lambda_{11})$ and $Y \sim \Ps(\lambda_{01}+\lambda_{11})$. The odds ratios (\ref{eqn:omxy}) reduce down to
\[\omega_{xy} = \sum_{i=0}^{\min(x,y)} i! \binom{x}{i} \binom{y}{i}  \left(\frac{\lambda_{11}}{\lambda_{10} \lambda_{01}}\right)^i , \qquad (x,y) \in \N_+ \times \N_+.\]
It is seen that the dependence structure in such a bivariate Poisson vector only depends on the parameter $\omega \doteq \lambda_{11}/(\lambda_{10} \lambda_{01})$. If the bivariate Poisson distribution is understood as a limiting version of a bivariate Binomial \citep[Section 4]{Marshall85}, then this $\omega$ would indeed be akin to the odds ratio in the constituting initial bivariate Bernoulli distribution.

\ppn Acting as in the previous section, one can first truncate $X$ and $Y$ at $N-1$, for obtaining discrete copula pmf's and then let $N$ tend to infinity for obtaining the Poisson copula densities shown in Figure \ref{fig:Poisscop1} for $\omega = 0.01$ and Figure \ref{fig:Poisscop001} for $\omega = 0.2$, respectively. 



\begin{figure}[h] 
\begin{subfigure}[t]{0.5\textwidth}
\centering
\includegraphics[width=0.8\textwidth]{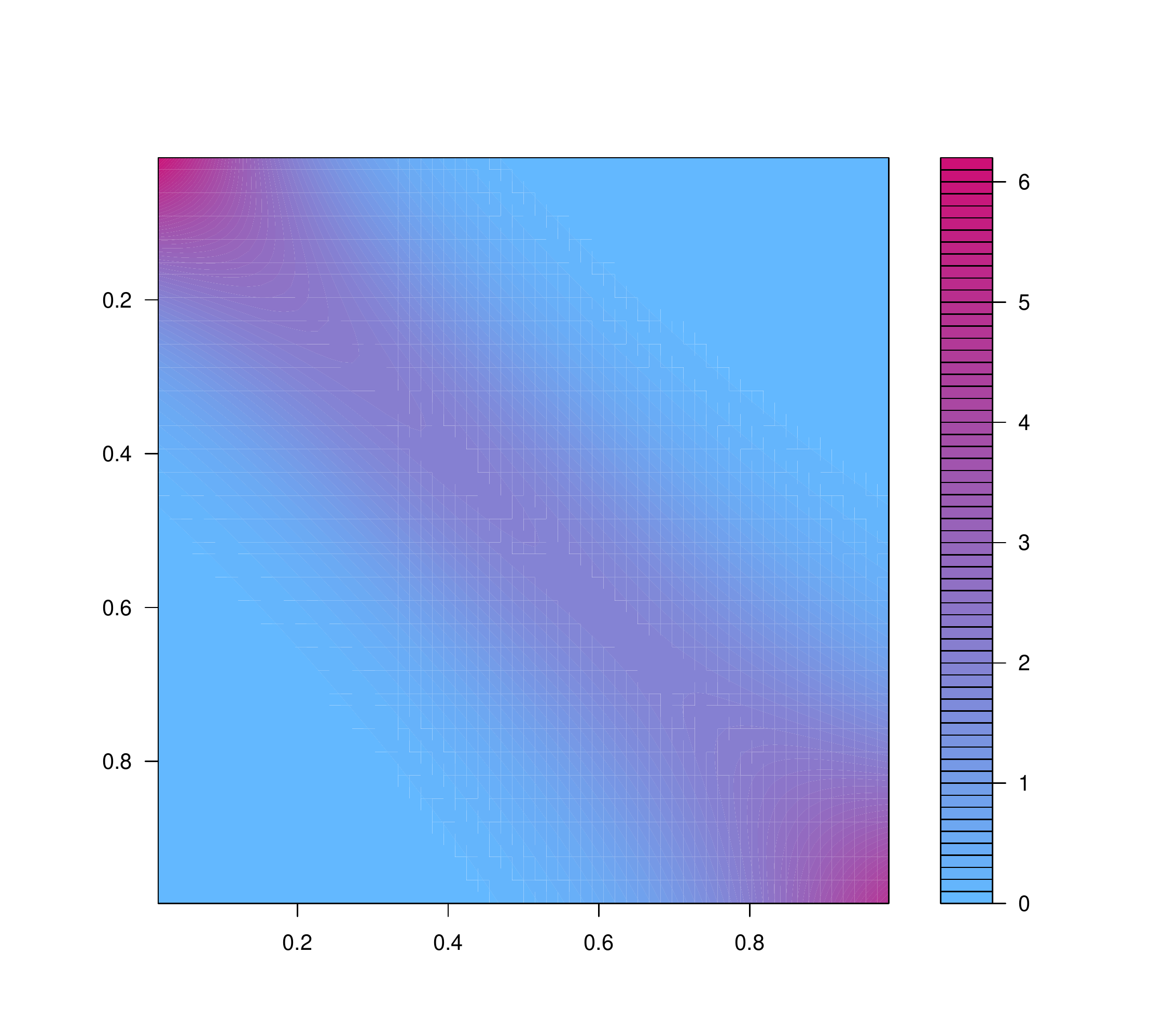}
\caption{Heat map}
\label{fig:Poisscop001cont}
\end{subfigure}
\hfill
\begin{subfigure}[t]{0.5\textwidth}
\centering
\includegraphics[width=0.9\textwidth]{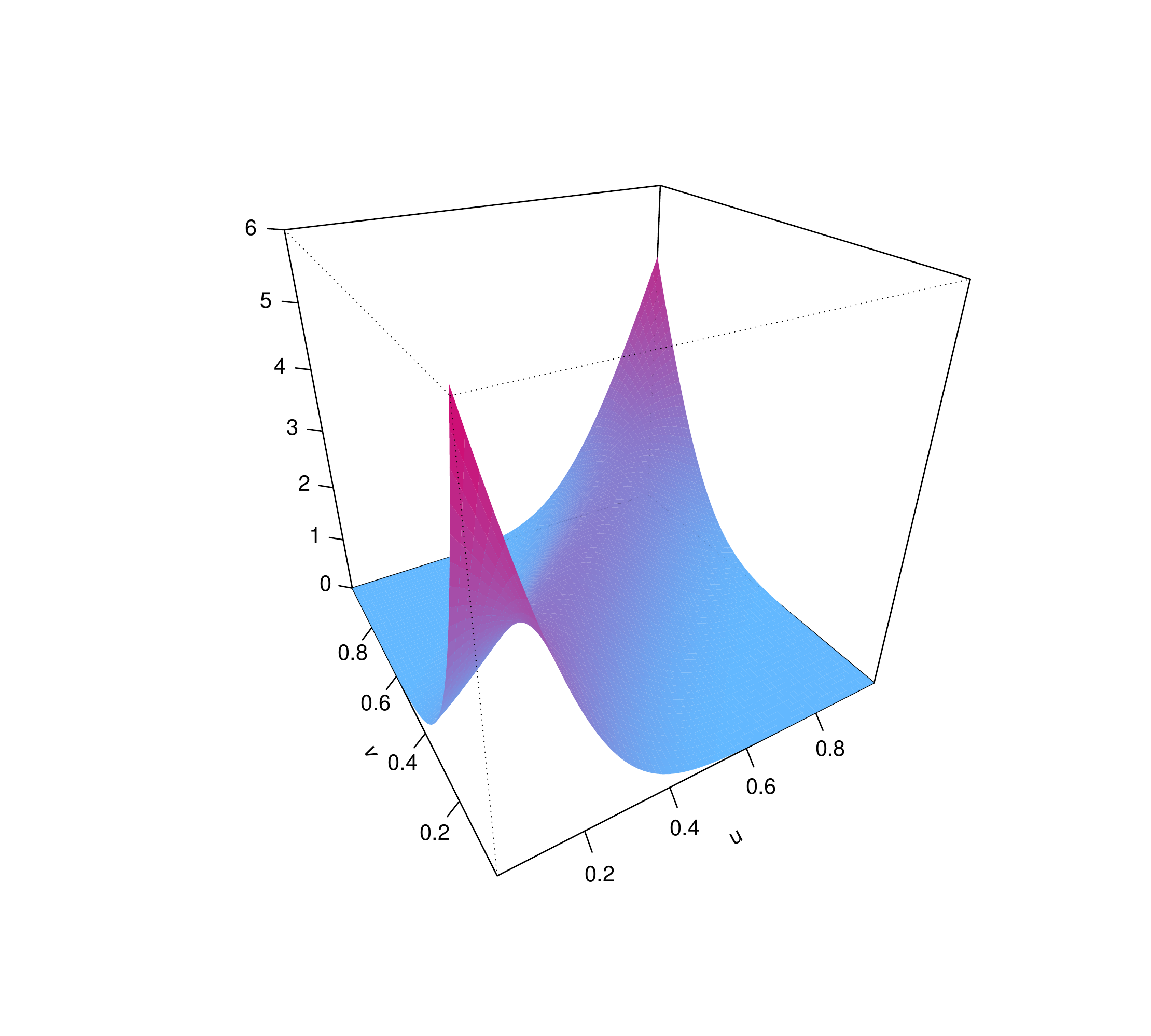}
\caption{3D-density estimate}
\label{fig:Poisscop001persp}
\end{subfigure}
\caption{The Poisson copula density with $\omega=0.01$.}
\label{fig:Poisscop001}
\end{figure}

\begin{figure}[h] 
\begin{subfigure}[t]{0.5\textwidth}
\centering
\includegraphics[width=0.8\textwidth]{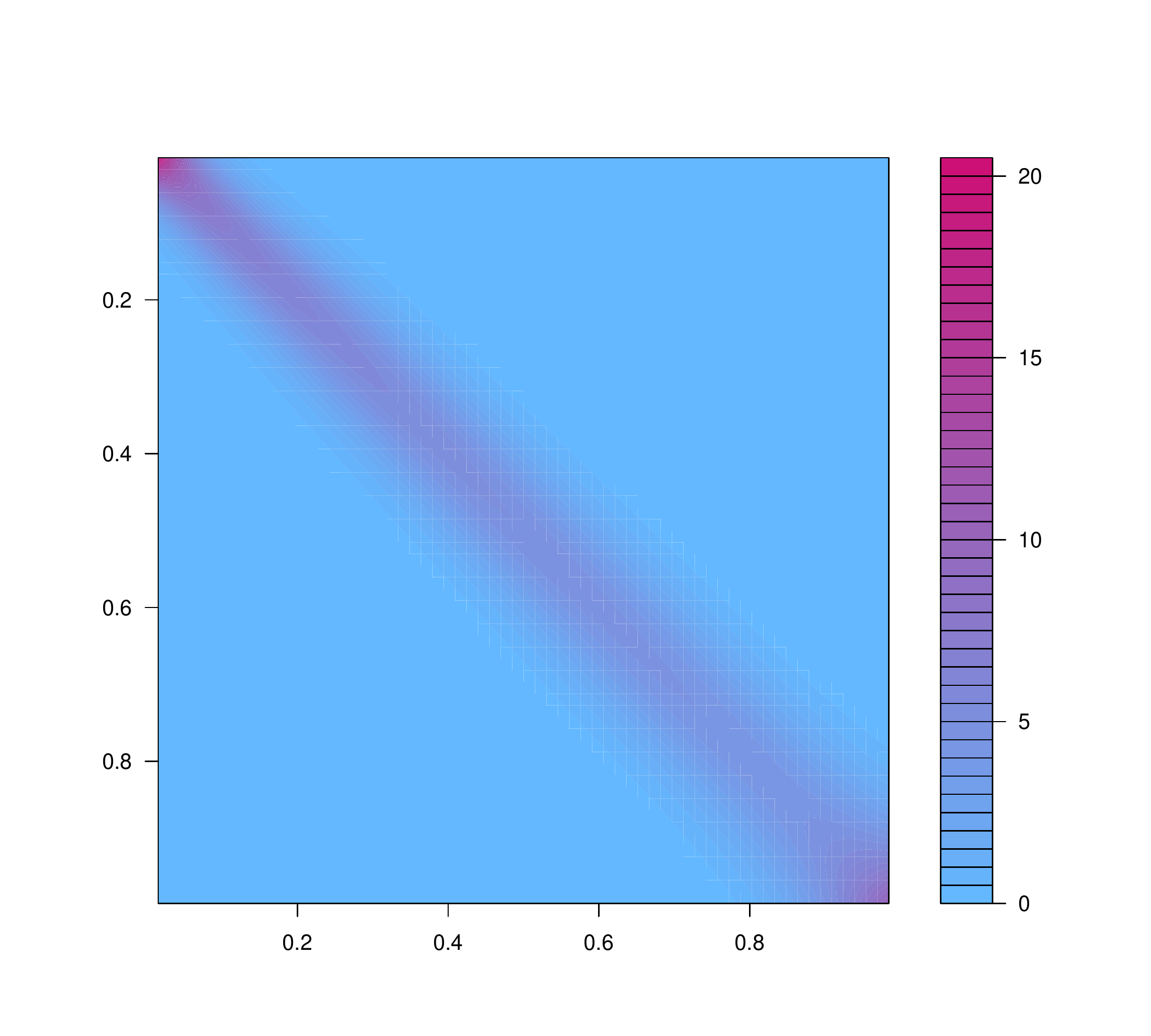}
\caption{Heat map}
\label{fig:Poisscop1cont}
\end{subfigure}
\hfill
\begin{subfigure}[t]{0.5\textwidth}
\centering
\includegraphics[width=0.9\textwidth]{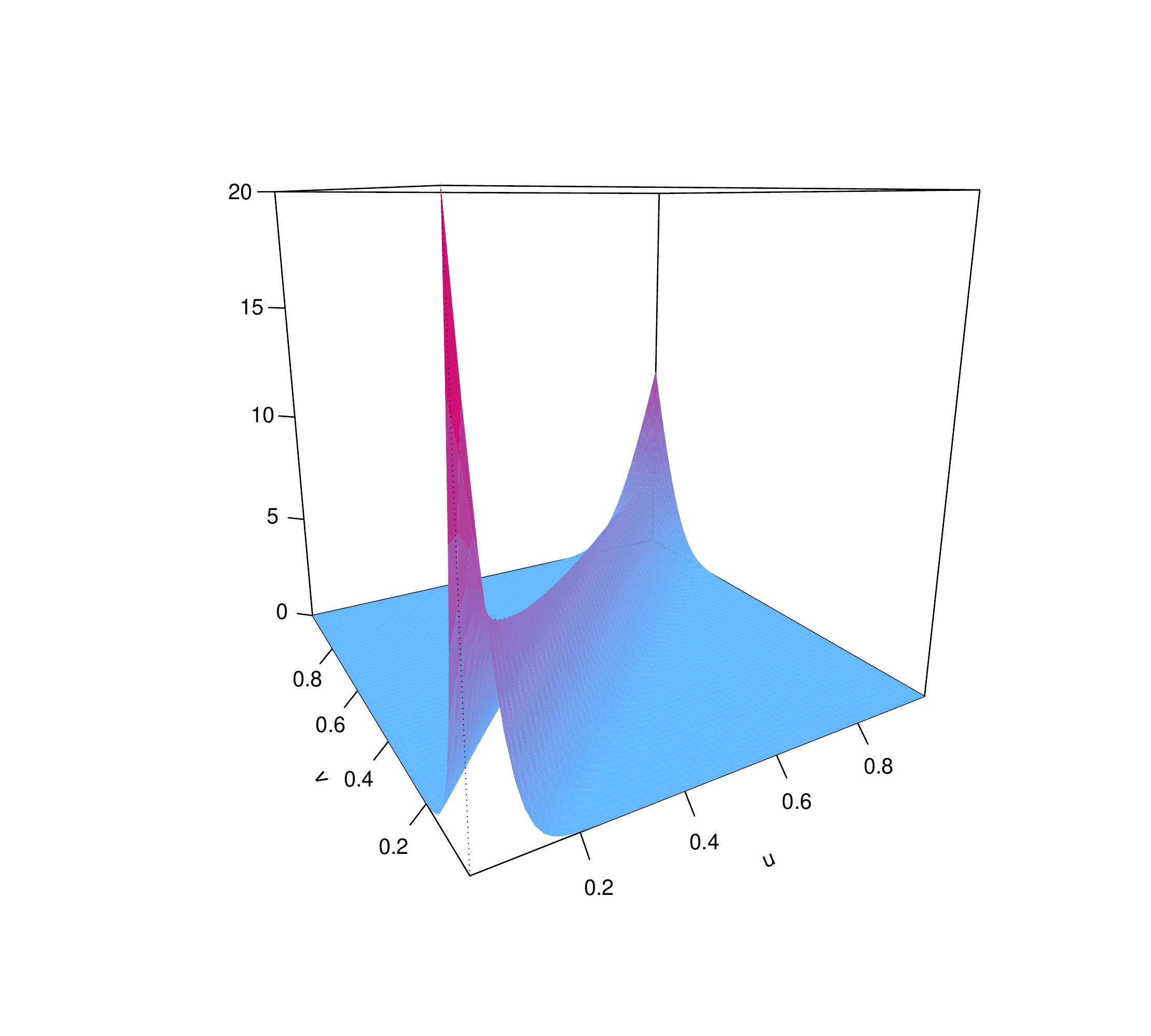}
\caption{3D-density estimate}
\label{fig:Poisscoppersp1}
\end{subfigure}
\caption{The Poisson copula density with $\omega=0.2$.}
\label{fig:Poisscop1}
\end{figure}

\ppn Like in any bivariate discrete distribution built on such an idea of `trivariate reduction' (\ref{eqn:trivred}), the components $X$ and $Y$ of a bivariate Poisson vector can only show positive association. How to construct bivariate discrete distributions with Poisson marginals showing negative association has been a challenging problem for a long time. For instance, \cite{Griffiths79} noted: ``{\it we have been unable to discover explicit in the literature any examples of bivariate Poisson distributions in which the correlation is negative}.'' However, they continued: ``{\it though } [...] {\it such examples are implicit in work of Fr\'echet (1951) and Hoeffding (1940)}'', these references obviously being part of the early literature on copulas. Indeed a systematic classical copula construction, based on (\ref{eqn:Sklar}), has been proposed in \cite{Pfeifer04}. Following the discussion in Section \ref{sec:copfordiscr}, in particular, the impossibility of ever disjointing margins and dependence structure, such construction should be subject to caution.

\ppn By contrast, it is easy to couple any two Poisson distributions with any continuous copula through IPF (Sections \ref{sec:IPF}-\ref{sec:arbdistrRS}). Figure \ref{fig:BivPoissNeg} shows confetti plots of three bivariate discrete distributions with Poisson $\Ps(2)$ marginals and negative association; coupled through $(a)$ a Clayton copula with $\theta = -0.2$; $(b)$ a Gaussian copula with $\rho = -0.8$; and $(c)$ a Geometric copula with $\omega = 1/2$ (Figure \ref{fig:geomcoptrunc2}). This illustrates that the discrete copula approach proposed in this paper shares with its continuous counterpart the same flexibility for constructing `new' bivariate distributions with arbitrary marginals and arbitrary dependence structure.

\begin{figure}[h]
\centering
\includegraphics[width=\textwidth]{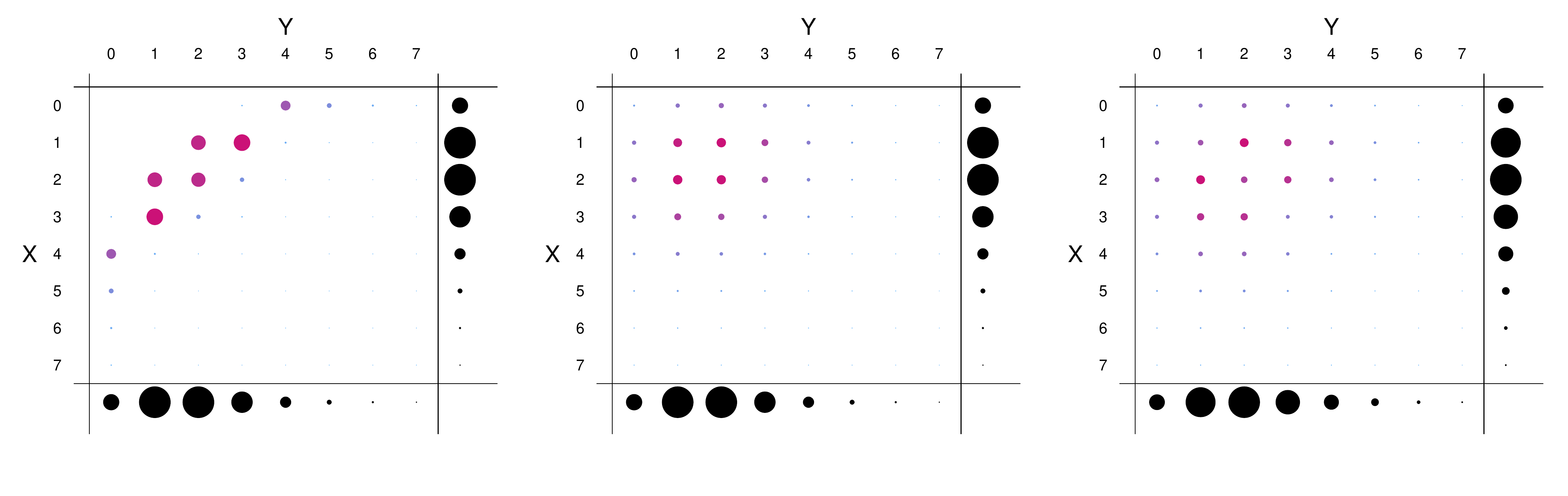}
\caption{Confetti plots of three bivariate discrete distributions with $\Ps(2)$-margins coupled by $(a)$ a Clayton copula with $\theta = -0.2$ (left); $(b)$ a Gaussian copula with $\rho = -0.8$ (middle) and $(c)$ a Geometric copula with $\omega = 1/2$ (right). All three show negative association between $X$ and $Y$.}
\label{fig:BivPoissNeg}
\end{figure}

\section{Concluding remarks} \label{sec:ccl}

The classical definition of a copula (Definition \ref{dfn:classcop}) follows implicitly but directly from the Probability Integral Transform (PIT). Hence it is fundamentally grounded in the continuous framework, and there is little surprise that classical copula ideas lead to many inconsistencies when applied on discrete random vectors. What is surprising is that a large part of the previous literature in the field has tried to make such an inherently continuous concept forcibly fit the discrete case as well, in spite of those inconsistencies.

\ppn In this paper it is argued that the very essence of a copula, understood as the `glue' between the marginals in a bivariate distribution, has nothing to do with PIT or uniform distributions, and should not be imprisoned in Definition \ref{dfn:classcop}. Rather, a copula is akin to an equivalence class of distributions sharing the same dependence structure. Defining such equivalence classes, called {\it nuclei}, does not require resorting to PIT and hence smoothly carries over to the discrete case. This paper describes that `discrete copula' construction. It is seen that all the pleasant properties of copulas for modelling dependence are maintained in the presented discrete framework, such as margin-freeness of anything copula-based or flexibility in constructing new bivariate distributions with arbitrary marginals and dependence structure without interaction between the two, as opposed to when classical copulas are naively applied to discrete vectors. 

\ppn Theoretical results on the existence and uniqueness of the copula probability mass function (copula pmf), analogue to the copula density in the continuous case, are obtained. The ideas are first introduced in the bivariate Bernoulli case, i.e., distributions supported on $\{0,1\} \times \{0,1\}$, and then gradually generalised to distributions supported on $\{0,1,\ldots,R\} \times \{0,1,\ldots,S\}$, for some finite $R$ and $S$, and finally to bivariate distributions supported on $\N \times \N$. Interestingly, the dependence structure in such a $(\N \times \N)$-supported distribution may still be captured by a classical continuous copula, and that copula is unique. However, that copula is {\it not} one of the copulas $C$ appearing in Sklar's theorem (\ref{eqn:Sklar}), as those inherently rely on the marginal distributions in direct contradiction with the initial motivation behind copula modelling. The construction gives rise to `new' continuous copulas, such as the Geometric copula, representing the dependence structure in Marshall and Olkin's bivariate Geometric distribution, or the Poisson copula, describing the dependence within a bivariate Poisson random vector. Purely discrete copulas are also introduced, such as the Bernoulli copula, the Binomial copula or the Goodman copula.

\ppn The whole methodology presented in this paper is largely inspired by century-old ideas put forward by Udny Yule in the first place, the `old bottle'. Yet it remains in total agreement with Sklar's theorem, the `new wine', although it challenges some of the ways it has been interpreted some times. 

\appendix
\section*{Appendix}
\renewcommand\thefigure{A.\arabic{figure}}    

\begin{figure}[h]
\centering
\includegraphics[width=0.6\textwidth]{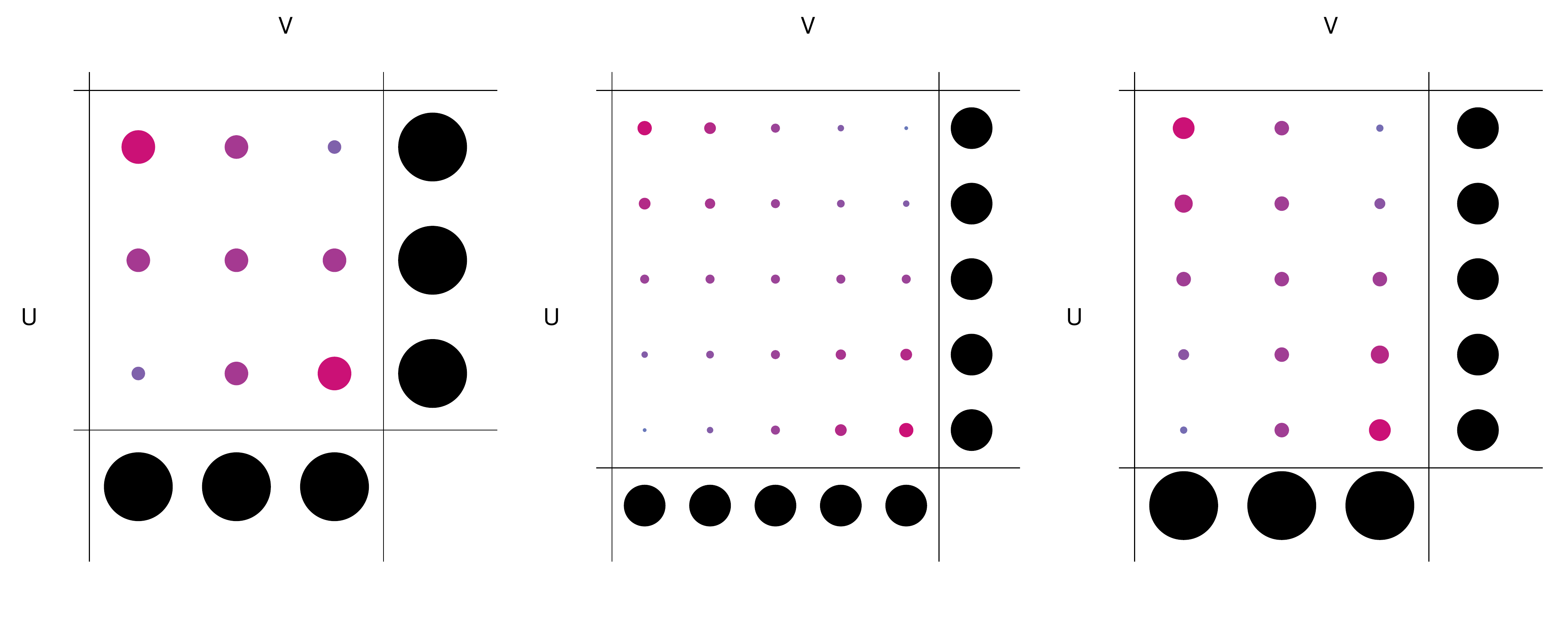}
\caption{Confetti plots of the FGM copula pmf with $\theta = 1$ and $(R,S) = (3,3)$ (left), $(R,S) = (5,5)$ (middle) and $(R,S) = (5,3)$ (right).}
\label{fig:FGMcop}
\end{figure}

\begin{figure}[h]
\centering
\includegraphics[width=0.6\textwidth]{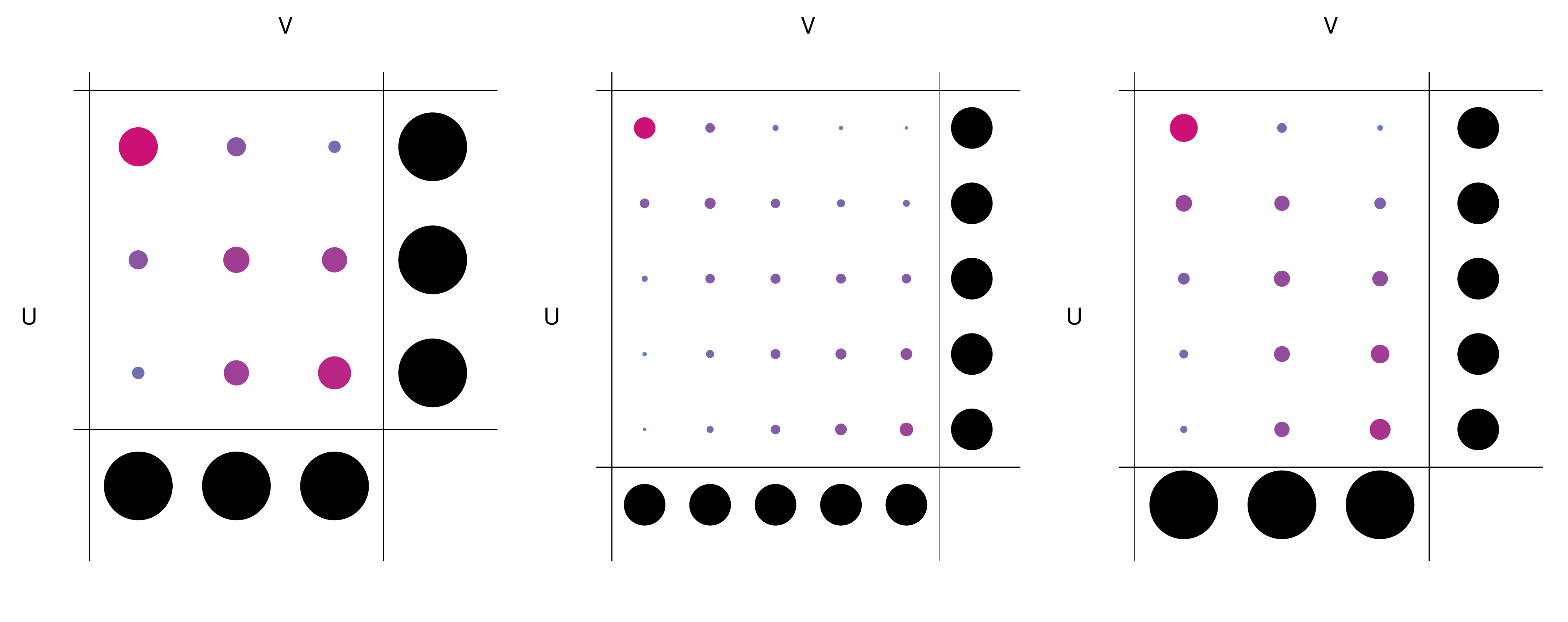}
\caption{Confetti plots of the Clayton copula pmf with $\theta = 0.8$ and $(R,S) = (3,3)$ (left), $(R,S) = (5,5)$ (middle) and $(R,S) = (5,3)$ (right).}
\label{fig:Claytposcop}
\end{figure}

\begin{figure}[h]
\centering
\includegraphics[width=0.6\textwidth]{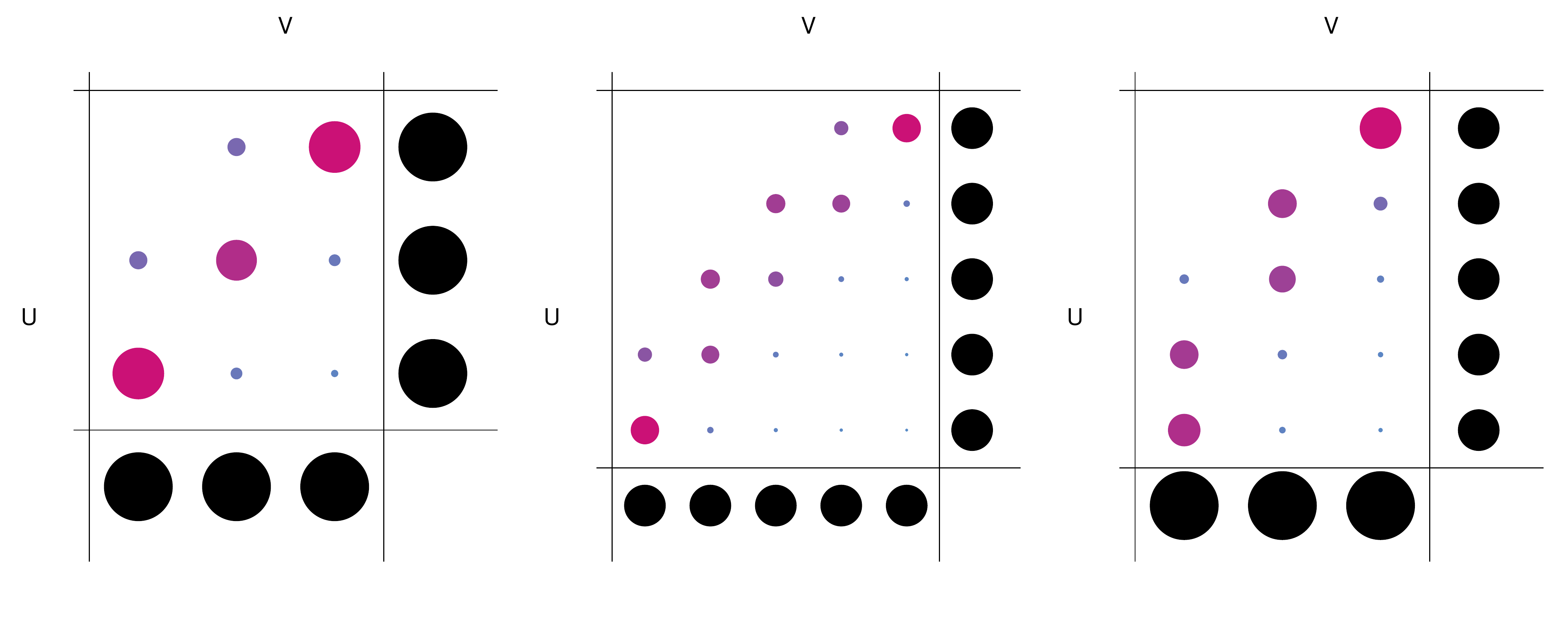}
\caption{Confetti plots of the Clayton copula pmf with $\theta = -0.8$ and $(R,S) = (3,3)$ (left), $(R,S) = (5,5)$ (middle) and $(R,S) = (5,3)$ (right).}
\label{fig:Claytnegcop}
\end{figure}

\begin{figure}[h]
\centering
\includegraphics[width=0.6\textwidth]{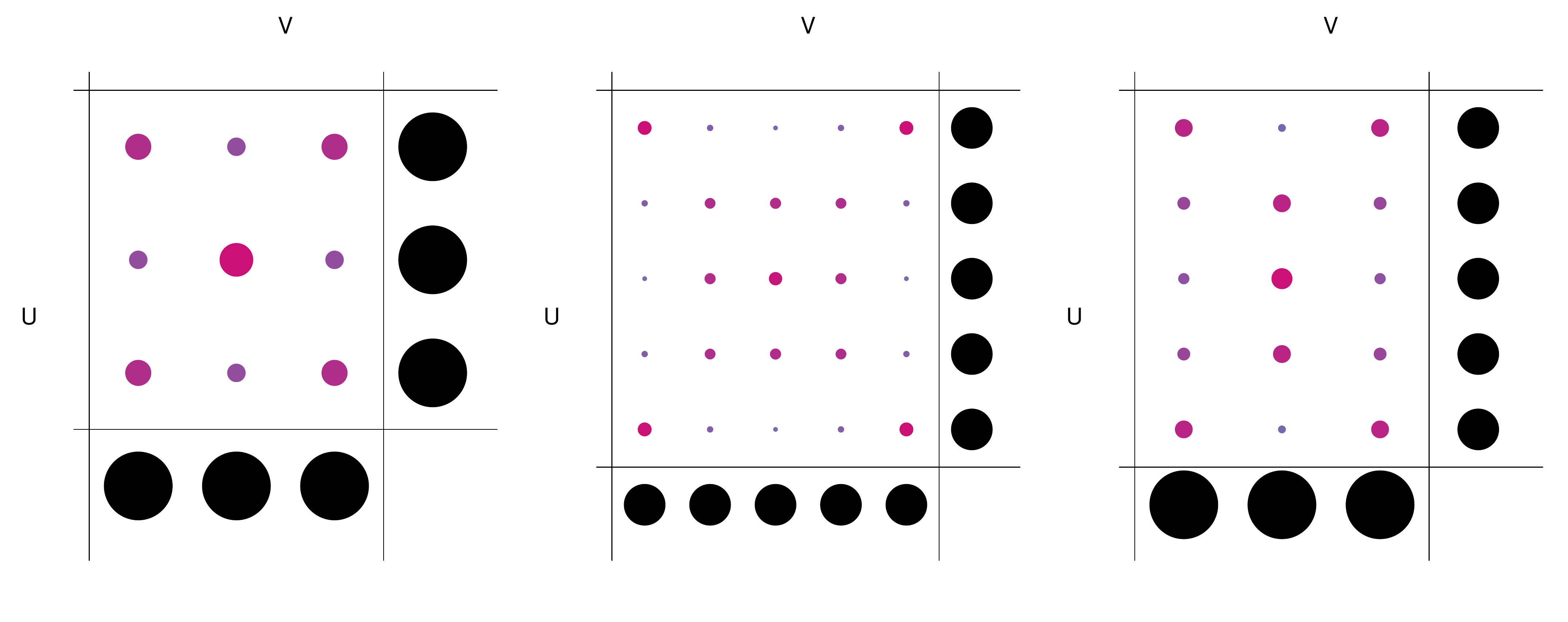}
\caption{Confetti plots of the Student copula pmf with $\rho = 0$, $\text{df}=1$ and $(R,S) = (3,3)$ (left), $(R,S) = (5,5)$ (middle) and $(R,S) = (5,3)$ (right).}
\label{fig:Stucop}
\end{figure}

\begin{figure}[h]
\centering
\includegraphics[width=0.6\textwidth]{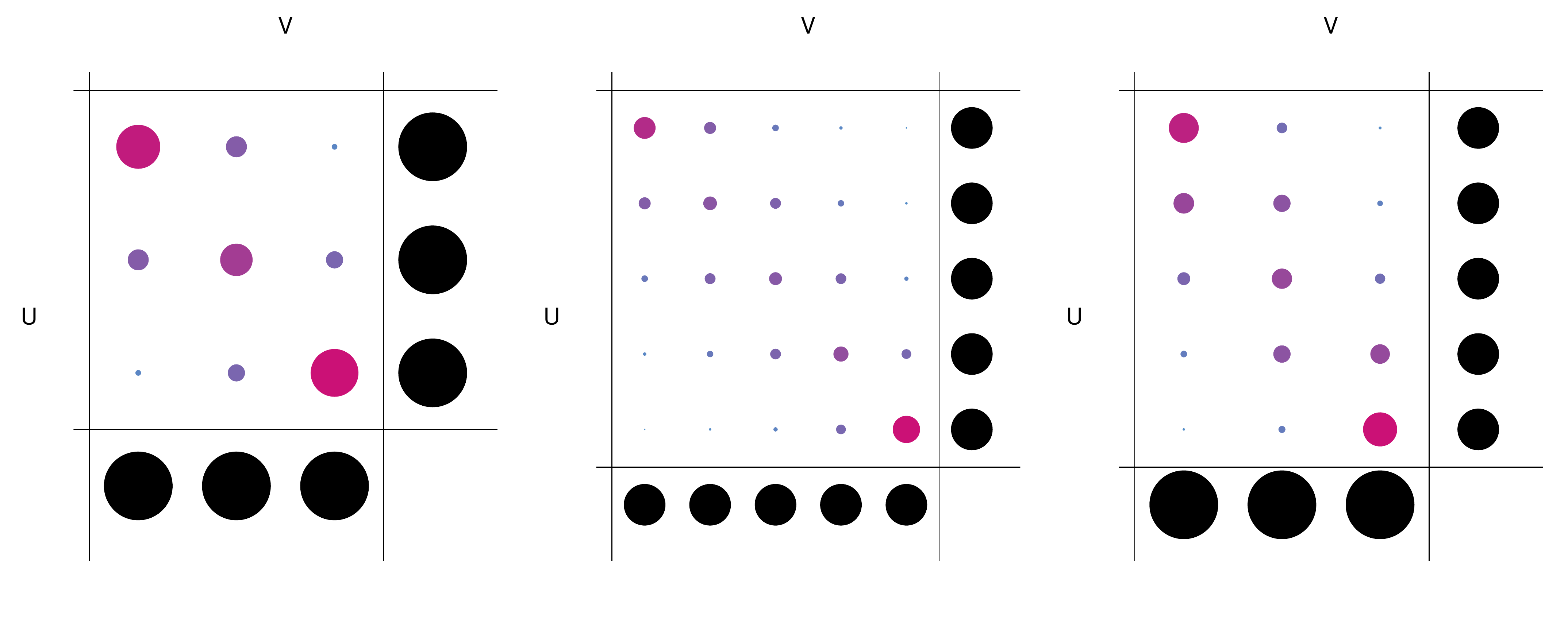}
\caption{Confetti plots of the Gumbel copula pmf with $\theta = 2$ and $(R,S) = (3,3)$ (left), $(R,S) = (5,5)$ (middle) and $(R,S) = (5,3)$ (right).}
\label{fig:Gumbcop}
\end{figure}

\end{document}